\documentclass{article}

\usepackage{cmbright}
\usepackage{amssymb,amsmath,amsthm}
\usepackage{mathrsfs,mathabx}
\usepackage[colorlinks,allcolors=blue]{hyperref}
\usepackage{graphicx}

\def\B{\mathscr B}
\def\C{\mathbb C}
\def\d{\mathrm d}
\def\D{\mathscr D}
\def\dom{\mathcal D}
\def\E{\mathscr E}
\def\G{\mathcal G}
\def\h{\mathfrak h}
\def\H{\mathcal H}
\def\K{\mathscr K}
\def\ltwo{\mathop{\mathrm{L}^2}\nolimits}
\def\ltwoloc{\mathop{\mathrm{L}^2_{\rm loc}}\nolimits}
\def\linfty{\mathop{\mathrm{L}^\infty}\nolimits}
\def\N{\mathbb N}
\def\NN{\mathfrak N}
\def\O{\mathcal O}
\def\R{\mathbb R}
\def\S{\mathscr S}
\def\Tau{\mathcal T}
\def\U{\mathscr U}
\def\Z{\mathbb Z}

\def\e{\mathop{\mathrm{e}}\nolimits}
\def\im{\mathop{\mathrm{Im}}\nolimits}
\def\Ker{\mathop{\mathrm{Ker}}\nolimits}
\def\Ran{\mathop{\mathrm{Ran}}\nolimits}
\DeclareMathOperator*{\slim}{s\hspace{0.1pt}-\hspace{0.1pt}lim}
\def\supp{\mathop{\mathrm{supp}}\nolimits}

\newtheorem{Theorem}{Theorem}[section]
\newtheorem{Remark}[Theorem]{Remark}
\newtheorem{Lemma}[Theorem]{Lemma}
\newtheorem{Corollary}[Theorem]{Corollary}
\newtheorem{Proposition}[Theorem]{Proposition}
\newtheorem{Definition}[Theorem]{Definition}
\newtheorem{Assumption}[Theorem]{Assumption}

\begin{document}


\title{Spectral and scattering theory of one-dimensional coupled photonic crystals}

\author{G. De Nittis$^1$\footnote{Supported by the Chilean Fondecyt Grant 1190204.},
M. Moscolari$^2$\footnote{Supported by the National Group of Mathematical Physics
(GNFM-INdAM)}, S. Richard$^3$\footnote{Supported by the grant\emph{Topological
invariants through scattering theory and noncommutative geometry} from Nagoya
University, and by JSPS Grant-in-Aid for scientific research C no 18K03328, and on
leave of absence from Univ.~Lyon, Universit\'e Claude Bernard Lyon 1, CNRS UMR 5208,
Institut Camille Jordan, 43 blvd.~du 11 novembre 1918, F-69622 Villeurbanne cedex,
France.}, R. Tiedra de Aldecoa$^1$\footnote{Supported by the Chilean Fondecyt Grant
1170008.}}

\date{\small}
\maketitle
\vspace{-1cm}

\begin{quote}
\emph{
\begin{enumerate}
\item[$^1$] Facultad de Matem\'aticas \& Instituto de F\'isica, Pontificia Universidad
Cat\'olica de Chile,\\ Av. Vicu\~na Mackenna 4860, Santiago, Chile
\item[$^2$] Fachbereich Mathematik, Eberhard-Karls-Universit\"at, Auf der Morgenstelle 10, 72076 T\"ubingen, Germany
\item[$^3$] Graduate school of mathematics, Nagoya University,
Chikusa-ku,\\Nagoya 464-8602, Japan
\item[]\emph{E-mails:} gidenittis@mat.uc.cl, massimo.moscolari@mnf.uni-tuebingen.de,
richard@math.nagoya-u.ac.jp, rtiedra@mat.puc.cl
\end{enumerate}
}
\end{quote}


\begin{abstract}
We study the spectral and scattering theory of light transmission in a system
consisting of two asymptotically periodic waveguides, also known as one-dimensional
photonic crystals, coupled by a junction. Using analyticity techniques and commutator
methods in a two-Hilbert spaces setting, we determine the nature of the spectrum and
prove the existence and completeness of the wave operators of the system.
\end{abstract}


\smallskip


\tableofcontents

\section{Introduction and main results}\label{sec_intro}
\setcounter{equation}{0}

In this paper, we study the propagation of an electromagnetic field $(\vec E,\vec H)$
in an infinite one-dimensional waveguide. We assume that (i) the waveguide is parallel
to the $x$-axis of the Cartesian coordinate system; (ii) the  electric field varies
along the $y$-axis and is constant on the planes perpendicular to the $x$-axis, i.e.,
$\vec E(x,y,z,t)=\varphi_E(x,t)\widehat y$; (iii) the  magnetic field varies along the
$z$-axis and is constant on the planes perpendicular to the $x$-axis, i.e.,
$\vec H(x,y,z,t)=\varphi_H(x,t)\widehat z$; (iv)  the waveguide is made of
isotropic medium\footnote{The interaction between the electromagnetic field and the
dielectric medium is characterised by the electric permittivity tensor $\varepsilon$
and the magnetic permeability tensor $\mu$. In an isotropic medium these tensors are
multiple of the identity, and thus determined by two scalars.}. Under these
assumptions, one has $\nabla\times\vec E=(\partial_x\varphi_E)\widehat z$ and
$\nabla\times\vec H=(-\partial_x\varphi_H)\widehat y$ and the dynamical sourceless
Maxwell equations \cite{Jac99} read as
\begin{equation}\label{eq_Maxwell_iso}
\begin{cases}
\varepsilon\;\!\partial_t\varphi_E=-\partial_x\varphi_H\\
\mu\;\!\partial_t\varphi_H=-\partial_x\varphi_E.
\end{cases}
\end{equation}
The scalar quantities $\varepsilon$ and $\mu$ in \eqref{eq_Maxwell_iso} are the
electric permittivity and magnetic permeability, respectively. They are strictly
positive functions on $\R$ describing the interaction of the waveguide with the
electromagnetic field. One can also include in the model effects associated to
bi-anisotropic media \cite{LSTV94}. In our case, this is achieved by modifying
the system \eqref{eq_Maxwell_iso} as follows\;\!:
\begin{equation}\label{eq_Maxwell_bi_iso}
\begin{cases}
\varepsilon\;\!\partial_t\varphi_E+\chi\;\!\partial_t\varphi_H=-\partial_x\varphi_H\\
\mu\;\!\partial_t\varphi_H+\chi^*\partial_t\varphi_E=-\partial_x\varphi_E.
\end{cases}
\end{equation}
The (possibly complex-valued) function $\chi$ is called bi-anisotropic coupling
term\footnote{In the general theory of bi-anisotropic media, $\chi$ is a tensor rather
than a scalar. The system of equations \eqref{eq_Maxwell_bi_iso} corresponds to a
particular choice of the form of this tensor. For more details on the theory of
bi-anisotropic media, we refer the interested reader to the monograph \cite{LSTV94}.}.
In the sequel, we will refer to the triple $(\varepsilon,\mu,\chi)$ as the
\emph{constitutive} functions of the waveguide.

Let us first discuss the case of periodic waveguides, also known as one-dimensional
photonic crystals, consisting in one-dimensional media with dielectric properties
which vary periodically in space \cite{JJW99, PTB75,ZRN+12}. Mathematically, this
translates into the fact that the functions $\varepsilon,\mu$ and $\chi$ in
\eqref{eq_Maxwell_bi_iso} are periodic, all with the same period. This makes
\eqref{eq_Maxwell_bi_iso} into a coupled system of differential equations with
periodic coefficients, and standard techniques like Bloch-Floquet theory (see e.g.
\cite{Kuc93}) can be used to study the propagation of solutions (or modes). One of the
fundamental properties of periodic waveguides is the presence of a frequency spectrum
made of bands and gaps. This implies that not all the modes can propagate along the
medium, since the propagation of modes associated to frequencies inside a gap is
forbidden by the ``geometry'' of the system. This phenomenon is similar to the one
appearing in the theory of periodic Schr\"odinger operators, where one has electronic
energy bands instead of frequency bands \cite[Sec.~XIII.16]{RS4}.

The study of the propagation of light in a periodic waveguide can be performed using
Bloch-Floquet theory. The situation becomes more complicated when one wants to study
the propagation of light through two periodic waveguides of different periods that are
connected by a junction. Such a system is schematically  represented in Figure
\ref{figure}. The asymptotic behaviour of the system on the left is characterised by
the periodic constitutive functions $(\varepsilon_\ell,\mu_\ell,\chi_\ell)$, whereas
the asymptotic behaviour of the system on the right is characterised by periodic
constitutive functions $(\varepsilon_{\rm r},\mu_{\rm r},\chi_{\rm r})$. Namely,
$\varepsilon\to\varepsilon_\ell$ when $x\to-\infty$ and
$\varepsilon\to\varepsilon_{\rm r}$ when $x\to+\infty$, and similarly for the other
two functions $\mu$ and $\chi$ (see Assumption \ref{ass_weight} for a precise
statement). The full system represented in Figure \ref{figure} can therefore be
interpreted as a perturbation of a ``free'' system obtained by glueing together two
purely periodic systems, one with periodicity of type $\ell$ on the left and the other
with periodicity of type $\rm r$ on the right. Accordingly, the analysis of the
dynamics of the full system can be performed with the tools of spectral and scattering
theories, leading us exactly to the main goal of this work\;\!: \emph{the spectral and
scattering analysis of one-dimensional coupled photonic crystals.}

\begin{figure}[htbp]\label{figure}
\begin{center}
\includegraphics[width=370pt]{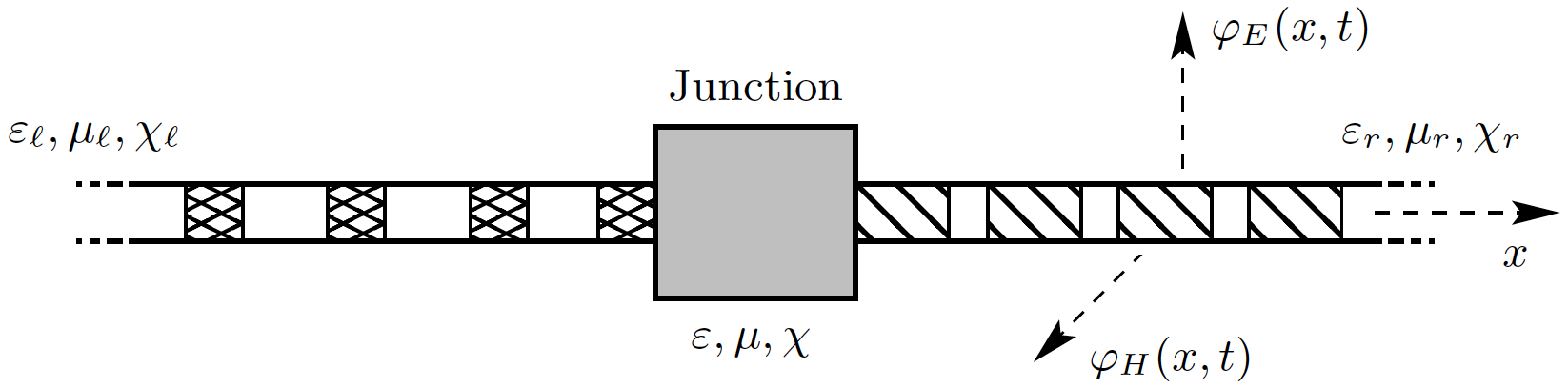}
\caption{Two periodic waveguides (one-dimensional photonic crystals) connected by a
junction}
\end{center}
\end{figure}

Since quantum mechanics provides a rich toolbox for the study of problems associated
to Schr\"odinger equations, we recast our equations of motion in a Schr\"odinger form
to take advantage of these tools, in particular commutator methods which will be used
extensively in this paper. Namely, with the notation
$
w:=\big(\begin{smallmatrix}\varepsilon & \chi\\\chi^* & \mu\end{smallmatrix}\big)^{-1}
$
for the positive-definite matrix of weights associated to the constitutive functions
$\varepsilon$, $\mu$, $\chi$ (Maxwell weight for short), we rewrite the system of
equations \eqref{eq_Maxwell_bi_iso} in the matrix form
\begin{equation}\label{eq_Maxwell_recast}
i\partial_t\begin{pmatrix}\varphi_E\\\varphi_H\end{pmatrix}
=w\begin{pmatrix}0 & -i\partial_x\\-i\partial_x & 0\end{pmatrix}
\begin{pmatrix}\varphi_E\\\varphi_H\end{pmatrix},
\end{equation}
so that it can be considered as a Schr\"odinger equation for the state
$(\varphi_E,\varphi_H)^{\rm T}$ in the Hilbert space $\ltwo(\R,\C^2)$. This
observation is by no means new. Since the dawn of quantum mechanics, the founding
fathers were well-aware that the Maxwell equations in vacuum are relativistically
covariant equations for a massless spin-1 particle \cite[pp. 151 \& 198]{Wig39}.
Moreover, similar Schr\"odinger formulations have already been employed in the
literature to study the quantum scattering theory of electromagnetic waves and other
classical waves in homogeneous media \cite{BS87,Kat67,RS77,SW71,Wil66}, and to study
the propagation of light in periodic media \cite{DL14,DL17,FK97,Kuc01}, among other
things. However, to the best of our knowledge, the specific problem we want to tackle
in the present work has never been considered in the literature.

The papers \cite{RS77,SW71,Wil66} deal with the scattering theory of three-dimensional
electromagnetic waves in a homogeneous medium. In that setup, the constitutive tensors
$\varepsilon$, $\mu$, $\chi$ are asymptotically constant. In contrast, in our
one-dimensional setup the constitutive functions $\varepsilon$, $\mu$, $\chi$ are only
assumed to be asymptotically periodic. This introduces a significant complication and
novelty to the model, even though it has lower dimension than the three-dimensional
models. Also, several works dealing with the scattering theory of electromagnetic
waves are conducted under the simplifying assumption that $\chi=0$ (absence of
bi-anisotropic effects), an assumption that we do not make in the present work. The
papers \cite{Bro65,Wil64} deal with the transmission of the electric field and voltage
along lines, also called one-dimensional Ohmic conductors. Mathematically, this
problem is described by a system of differential equations similar to
\eqref{eq_Maxwell_bi_iso} or \eqref{eq_Maxwell_recast}. However, in these papers, the
constitutive quantities, namely the self-inductance and capacitance, are once again
assumed to be asymptotically constant, in contrast with our less restrictive
assumption of asymptotic periodicity. 
Finally, in the paper \cite{Yaf05}, almost
no restrictions are imposed on the asymptotic behavior of the constitutive functions,
but a stronger condition (invertibility) is imposed on the operator modelling
the junction. Here, we do not assume that this operator is
invertible or isometric, since we want to describe the scattering effects produced
by the introduction of the junction itself, without imposing unnecessary conditions
on the relation between the free dynamics without interface and the full
dynamics in presence of the interface (see Remark \ref{rem_junction} for more details).
Also, even though the results of that paper hold in any space dimension,
our results for one-dimensional photonic crystals are more detailed.

To conclude our overview of the literature, we point out that the dynamical equations
describing our model are common to other physical systems. This is for instance the
case of the equations describing the propagation of an Alfv\'en wave in a periodically
stratified stationary plasma \cite{Ali92}, the propagation of linearized water waves
in a periodic bottom  topography \cite{CGLS18}, or the propagation of harmonic
acoustic waves in periodic waveguides \cite{Bra94}. In consequence, the results of our
analysis here can be applied to all these models by reinterpreting in a appropriate
way the necessary quantities.

Here is a description of our results. In Section \ref{sec_full}, we introduce our
assumption on the Maxwell weight $w$ (Assumption \ref{ass_weight}) and we define the
full Hamiltonian $M$ in the Hilbert space $\H_w$ describing the one-dimensional
coupled photonic crystal. In Section \ref{sec_free}, we define the free Hamiltonian
$M_0$ in the Hilbert space $\H_0$ associated to $M$, and we define the operator
$J:\H_0\to\H_w$ modelling the junction depicted in Figure \ref{figure} (Definition
\ref{def_junction}). The operator $M_0$ is the direct sum of an Hamiltonian $M_\ell$
describing the periodic waveguide asymptotically on the left and an Hamiltonian
$M_{\rm r}$ describing the periodic waveguide asymptotically on the right. In Section
\ref{sec_fibering}, we use Bloch-Floquet theory to show that the asymptotic
Hamiltonians $M_\ell$ and $M_{\rm r}$ fiber analytically in the sense of G\'erard and
Nier \cite{GN98} (Proposition \ref{prop_fibered}). As a by-product, we prove that
$M_\ell$ and $M_{\rm r}$ do not possess flat bands, and thus have purely absolutely
continuous spectra (Proposition \ref{prop_spec_asymp}). The analytic fibration of
$M_\ell$ and $M_{\rm r}$ provides also a natural definition for the set $\Tau_M$ of
thresholds in the spectrum of $M$ (Eqs.
\eqref{eq_thresholds_star}-\eqref{eq_thresholds_M}). In section \ref{sec_com}, we
recall from \cite{ABG96,Sah97} the necessary abstract results on commutator methods
for self-adjoint operators. In section \ref{sec_conj_free}, we construct for each
compact interval $I\subset\R\setminus\Tau_M$ a conjugate operator $A_{0,I}$ for the
free Hamiltonian $M_0$ and use it to prove a limiting absorption principle for $M_0$
in $I$ (Theorem \ref{thm_Mourre_Free} and the discussion that follows). In Section
\ref{sec_conj_full}, we use the fact that $A_{0,I}$ is a conjugate operator for $M_0$
and abstract results on the Mourre theory in a two-Hilbert spaces setting
\cite{RT13_2} to show that the operator $A_I:=JA_{0,I}J^*$ is a conjugate operator for
$M$ (Theorem \ref{thm_regul}). In Section \ref{sec_spec}, we use the operator $A_I$ to
prove a limiting absorption principle for $M$ in $I$, which implies in particular that
in any compact interval $I\subset\R\setminus\Tau_M$ the Hamiltonian $M$ has at most
finitely many eigenvalues, each one of finite multiplicity, and no singular continuous
spectrum (Theorem \ref{thm_spec_M}). Using Zhislin sequences (a particular type of
Weyl sequences), we also show in Proposition \ref{proposition_ess_M} that $M$ and
$M_0$ have the same essential spectrum. In Section \ref{sec_scatt_two}, we recall
abstract criteria for the existence and the completeness of wave operators in a
two-Hilbert spaces setting. Finally, in Section \ref{sec_scatt_photonic}, we use these
abstract results in conjunction with the results of the previous sections to prove the
existence and the completeness of wave operators for the pair $(M_0,M)$ (Theorem
\ref{thm_wave_max}). We also give an explicit description of the initial sets of the
wave operators in terms of the asymptotic velocity operators for the Hamiltonians
$M_\ell$ and $M_{\rm r}$ (Proposition \ref{prop_V_star} \& Theorem
\ref{thm_initial_max}).

\section{Model}\label{sec_model}
\setcounter{equation}{0}

\subsection{Full Hamiltonian}\label{sec_full}

In this section, we introduce the full Hamiltonian $M$ that we will study. It is a
one-dimensional Maxwell-like operator describing perturbations of an anisotropic
periodic one-dimensional photonic crystal.

Throughout the paper, for any Hilbert space $\H$, we write
$\langle\;\!\cdot\;\!,\;\!\cdot\;\!\rangle_\H$ for the scalar product on $\H$,
$\|\cdot\|_\H$ for the norm on $\H$, $\B(\H)$ for the set of bounded operators on
$\H$, and $\K(\H)$ for the set of compact operators on $\H$. We also use the notation
$\B(\H_1,\H_2)$ (resp. $\K(\H_1,\H_2)$) for the set of bounded (resp. compact)
operators from a Hilbert space $\H_1$ to a Hilbert space $\H_2$.

\begin{Definition}[One-dimensional Maxwell-like operator]\label{def_model}
Let $0<c_0<c_1<\infty$ and take a Hermitian matrix-valued function
$w\in\linfty\big(\R,\B(\C^2)\big)$ such that $c_0\le w(x)\le c_1$ for a.e. $x\in\R$.
Let $P$ be the momentum operator in $\ltwo(\R)$, that is, $Pf:=-if'$ for each
$f\in\H^1(\R)$, with $\H^1(\R)$ the first Sobolev space on $\R$. Let
$$
D\varphi:=
\begin{pmatrix}
0 & P\\
P & 0
\end{pmatrix}
\varphi,
\quad\varphi\in\dom(D):=\H^1(\R,\C^2).
$$
Then the Maxwell-like operator $M$ in $\ltwo(\R,\C^2)$ is defined as
$$
M\varphi:=wD\varphi,\quad\varphi\in\dom(M):=\dom(D).
$$
\end{Definition}

The Maxwell weight $w$ that we consider converges at $\pm\infty$ to periodic functions
in the following sense\;\!:

\begin{Assumption}[Maxwell weight]\label{ass_weight}
There exist $\varepsilon>0$ and hermitian matrix-valued functions
$w_\ell,w_{\rm r}\in\linfty\big(\R,\B(\C^2)\big)$ periodic of periods
$p_\ell,p_{\rm r}>0$ such that
\begin{equation}\label{eq_ass_weight}
\begin{aligned}
\big\|w(x)-w_\ell(x)\big\|_{\B(\C^2)}
&\le{\rm Const.}\;\!\langle x\rangle^{-1-\varepsilon},\quad\hbox{a.e. $x<0$,}\\
\big\|w(x)-w_{\rm r}(x)\big\|_{\B(\C^2)}
&\le{\rm Const.}\;\!\langle x\rangle^{-1-\varepsilon},\quad\hbox{a.e. $x>0$,}
\end{aligned}
\end{equation}
where the indexes $\ell$ and r stand for ``left'' and ``right'',
 and $\langle x\rangle:=(1+|x|^2)^{1/2}$.
\end{Assumption}

\begin{Lemma}\label{lemma_self}
Let Assumption \ref{ass_weight} be satisfied.
\begin{enumerate}
\item[(a)] One has for a.e. $x\in\R$ the inequalities
$$
c_0\le w_\ell(x)\le c_1\quad\hbox{and}\quad c_0\le w_{\rm r}(x)\le c_1,
$$
with $c_0,c_1$ introduced in Definition \ref{def_model}.
\item[(b)] The sesquilinear form
$$
\langle\cdot,\cdot\rangle_{\H_w}:\ltwo(\R,\C^2)\times\ltwo(\R,\C^2)\to\C,
\quad(\psi,\varphi)\mapsto\big\langle\psi,w^{-1}\varphi\big\rangle_{\ltwo(\R,\C^2)},
$$
defines a new scalar product on $\ltwo(\R,\C^2)$, and we denote by $\H_w$ the space
$\ltwo(\R,\C^2)$ equipped with $\langle\;\!\cdot\;\!,\;\!\cdot\;\!\rangle_{\H_w}$.
Moreover, the norm of $\ltwo(\R,\C^2)$ and $\H_w$ are equivalent, and the claim
remains true if one replaces $w$ with $w_\ell$ or $w_{\rm r}$.
\item[(c)] The operator $M$ with domain $\dom(M):=\H^1(\R,\C^2)$ is self-adjoint in
$\H_w$.
\end{enumerate}
\end{Lemma}

\begin{proof}
Point (a) is a direct consequence of the assumptions on $w,w_\ell,w_{\rm r}$. Point
(b) follows from the bounds $c_0\le w(x),w_\ell(x),w_{\rm r}(x)\le c_1$ valid for a.e.
$x\in\R$. Point (c) can be proved as in \cite[Prop.~6.2]{DL18}.
\end{proof}

\subsection{Free Hamiltonian}\label{sec_free}

We now define the free Hamiltonian associated to the operator $M$. Due to the
anisotropy of the Maxwell weight $w$ at $\pm\infty$, it is convenient to define left
and right asymptotic operators
$$
M_\ell:=w_\ell D\quad\hbox{and}\quad M_{\rm r}:=w_{\rm r}D,
$$
with $w_\ell$ and $w_{\rm r}$ as in Assumption \ref{ass_weight}. Lemma
\ref{lemma_self}(c) implies that the operators $M_\ell$ and $M_{\rm r}$ are
self-adjoint in the Hilbert spaces $\H_{w_\ell}$ and $\H_{w_{\rm r}}$, with the same
domain $\dom(M_\ell)=\dom(M_{\rm r})=\dom(M)$. Then we define the free Hamiltonian as
the direct sum operator
$$
M_0:=M_\ell\oplus M_{\rm r}
$$
in the Hilbert space $\H_0:=\H_{w_\ell}\oplus\H_{w_{\rm r}}$. Since the free
Hamiltonian acts in the Hilbert space $\H_0$ and the full Hamiltonian acts in the
Hilbert space $\H_w$, we need to introduce an identification operator between the
spaces $\H_0$ and $\H_w:$

\begin{Definition}[Junction operator]\label{def_junction}
Let $j_\ell,j_{\rm r}\in C^\infty(\R,[0,1])$ be such that
$$
j_\ell(x):=
\begin{cases}
1 &\hbox{if~~$x\le-1$}\\
0 &\hbox{if~~$x\ge-1/2$}
\end{cases}
\qquad\hbox{and}\qquad
j_{\rm r}(x):=
\begin{cases}
0 &\hbox{if~~$x\le1/2$}\\
1 &\hbox{if~~$x\ge1$.}
\end{cases}
$$
Then $J:\H_0\to\H_w$ is the bounded operator defined by
$$
J(\varphi_\ell,\varphi_{\rm r})
:=j_\ell\;\!\varphi_\ell
+j_{\rm r}\;\!\varphi_{\rm r},
$$
with adjoint $J^*:\H_w\to\H_0$ given by
$
J^*\varphi
=\big(w_\ell w^{-1}j_\ell\;\!\varphi,w_{\rm r}w^{-1}j_{\rm r}\;\!\varphi\big).
$
\end{Definition}

\begin{Remark}\label{rem_junction}
We call $J$ the junction operator because it models mathematically the junction
depicted in Figure \ref{figure}. Indeed, the Hamiltonian $M_0$ only describes the free
dynamics of the system in the bulk asymptotically on the left and in the bulk
asymptotically on the right. Since $M_0$ is the direct sum of the operators $M_\ell$
and $M_{\rm r}$, the interface effects between the left and the right parts of the
system are not described by $M_0$ in any way. The role of the operator $J$ is thus to
map the free bulk states of the system belonging to the direct sum Hilbert space
$\H_0$ onto a \emph{joined} state belonging to the physical Hilbert space $\H_w$,
where acts the full Hamiltonian $M$ describing the interface effects.

Given a state $\varphi\in\H_w$, the square norm $E(\varphi):=\|\varphi\|_{\H_w}^2$ can
be interpreted as the total energy of the electromagnetic field
$\varphi\equiv(\varphi_E,\varphi_H)^{\rm T}$. A direct computation shows that the
total energy of a state $J(\varphi_\ell,\varphi_{\rm r})$ obtained by joining bulk
states $\varphi_\ell\in\H_\ell$ and $\varphi_{\rm r}\in\H_{\rm r}$ satisfies
$$
E\big(J(\varphi_\ell,\varphi_{\rm r})\big)
=E_\ell(\varphi_\ell)+E_{\rm r}(\varphi_r)
+E_{\rm interface}(\varphi_\ell,\varphi_{\rm r})
$$
with $E_\ell(\varphi_\ell):=\|j_\ell\;\!\varphi_\ell\|_{\H_{w_\ell}}^2$ and
$E_{\rm r}(\varphi_r):=\|j_{\rm r}\;\!\varphi_{\rm r}\|_{\H_{w_{\rm r}}}^2$ the total
energies of the field $j_\ell\varphi_\ell$ on the left and the field $j_r\varphi_r$
on the right, and with
$$
E_{\rm interface}(\varphi_\ell,\varphi_{\rm r})
:=\big\langle\varphi_\ell,j_\ell\big(w^{-1}-w_\ell^{-1}\big)
\;\!j_\ell\;\!\varphi_\ell\big\rangle_{\ltwo(\R,\C^2)}
+\big\langle\varphi_{\rm r},j_{\rm r}\big(w^{-1}-w_{\rm r}^{-1}\big)
\;\!j_{\rm r}\;\!\varphi_{\rm r}\big\rangle_{\ltwo(\R,\C^2)}
$$
the energy associated with the left and right external interfaces of the junction. In
particular, one notices that there is no contribution to the energy associated to the
central region $(-1/2,1/2)$ of the junction. This physical observation shows as a
by-product that the operator $J$ is neither invertible, nor isometric.
\end{Remark}

\subsection{Fibering of the free Hamiltonian}\label{sec_fibering}

In this section, we introduce a Bloch-Floquet (or Bloch-Floquet-Zak or
Floquet-Gelfand) transform to take advantage of the periodicity of the operators
$M_\ell$ and $M_{\rm r}$. For brevity, we use the symbol $\star$ to denote either the
index ``$\ell$'' or the index ``r''.

Let
$$
\Gamma_\star:=\big\{np_\star\mid n\in\Z\big\}\subset\R
$$
be the one-dimensional lattice of period $p_\star$ with fundamental cell
$Y_\star:=[-p_\star/2,p_\star/2]$, and let
$$
\Gamma^*_\star:=\big\{2\pi n/p_\star\mid n\in\Z\big\}\subset\R
$$
be the reciprocal lattice of $\Gamma_\star$ with fundamental cell
$Y^*_\star:=[-\pi/p_\star,\pi/p_\star]$. For each $t\in\R$, we define the translation
operator
$$
T_t:\ltwoloc(\R,\C^2)\to\ltwoloc(\R,\C^2),\quad\varphi\mapsto\varphi(\;\!\cdot\;\!-t).
$$
Using this operator, we can define the Bloch-Floquet transform of a $\C^2$-valued
Schwartz function $\varphi\in\S(\R,\C^2)$ as
$$
(\U_\star\varphi)(k,\theta)
:=\sum_{n\in\Z}\e^{-ik(\theta-np_\star)}\big(T_{np_\star}\varphi\big)(\theta),
\quad k,\theta\in\R.
$$
One can verify that $\U_\star\varphi$ is $p_\star$-periodic in the variable $\theta$,
$$
(\U_\star\varphi)(k,\theta+\gamma)=(\U_\star\varphi)(k,\theta),
\quad\gamma\in\Gamma_\star,
$$
and $2\pi/p_\star$-pseudo-periodic in the variable $k$,
$$
(\U_\star\varphi)(k+\gamma^*,\theta)=\e^{-i\theta\gamma^*}(\U_\star\varphi)(k,\theta),
\quad\gamma^*\in\Gamma^*_\star.
$$

Now, let $\h_\star$ be the Hilbert space obtained by equipping the set
$$
\big\{\varphi\in\ltwoloc(\R,\C^2)
\mid\hbox{$T_\gamma\varphi=\varphi$ for all $\gamma\in\Gamma_\star$}\big\}
$$
with the scalar product
$$
\langle\varphi,\psi\rangle_{\h_\star}
:=\int_{Y_\star}\d\theta\,\big\langle\varphi(\theta),
w_\star(\theta)^{-1}\psi(\theta)\big\rangle_{\C^2}.
$$
Since $\h_\star$ and $\ltwo(Y_\star,\C^2)$ are isomorphic, we shall use both
representations. Next, let $\tau:\Gamma_\star^*\to\B(\h_\star)$ be the unitary
representation of the dual lattice $\Gamma_\star^*$ on $\h_\star$ given by
$$
\big(\tau(\gamma^*)\varphi\big)(\theta):=\e^{i\theta\gamma^*}\varphi(\theta),
\quad\hbox{$\gamma^*\in\Gamma_\star^*$, $\varphi\in\h_\star$, a.e. $\theta\in\R$,}
$$
and let $\H_{\tau,\star}$ be the Hilbert space obtained by equipping the set
$$
\big\{u\in\ltwoloc(\R,\h_\star)\mid\hbox{$u(\;\!\cdot\;\!-\gamma^*)=\tau(\gamma^*)u$
for all $\gamma^*\in\Gamma^*_\star$}\big\}
$$
with the scalar product
$$
\langle u,v\rangle_{\H_{\tau,\star}}
:=\tfrac1{|Y^*_\star|}\int_{Y^*_\star}\d k\,\langle u(k),v(k)\rangle_{\h_\star}.
$$
There is a natural isomorphism from $\H_{\tau,\star}$ to $\ltwo(Y_\star^*,\h_\star)$
given by the restriction from $\R$ to $Y_\star^*$, and with inverse given by
$\tau$-equivariant continuation. However, using $\H_{\tau,\star}$ has various
advantages and we shall stick to it in the sequel. Direct calculations show that the
Bloch-Floquet transform extends to a unitary operator
$\U_\star:\H_{w_\star}\to\H_{\tau,\star}$ with inverse
$$
\big(\U_\star^{-1}u\big)(x)
=\tfrac1{|Y^*_\star|}\int_{Y^*_\star}\d k\,\e^{ikx}\big(u(k)\big)(x),
\quad\hbox{$u\in\H_{\tau,\star}$, a.e. $x\in\R$.}
$$
Furthermore, since $M_\star$ commutes with the translation operators $T_\gamma$
($\gamma\in\Gamma_\star$), the operator $M_\star$ is decomposable in the Bloch-Floquet
representation. Namely, we have
$$
\widehat{M_\star}
:=\U_\star M_\star\U_\star^{-1}
=\big\{\widehat{M_\star}(k)\big\}_{k\in\R}
$$
with
\begin{equation}\label{eq_cov}
\widehat{M_\star}(k-\gamma^*)=\tau(\gamma^*)\widehat{M_\star}(k)\tau(\gamma^*)^*,
\quad k\in\R,~\gamma^*\in\Gamma^*_\star,
\end{equation}
and
$$
\widehat{M_\star}(k)=w_\star\widehat D(k)
\quad\hbox{and}\quad
\widehat D(k)u(k)=
\begin{pmatrix}
0 & -i\partial_\theta+k\\
-i\partial_\theta+k & 0
\end{pmatrix}
u(k),\quad k\in Y_\star^*,~u\in\U_\star\;\!\dom(M_\star).
$$
Here, the domain $\U_\star\;\!\dom(M_\star)$ of $\U_\star M_\star\U_\star^{-1}$
satisfies
$$
\U_\star\;\!\dom(M_\star)
=\U_\star\;\!\H^1(\R,\C^2)
=\big\{u\in\ltwoloc(\R,\h_\star^1)\mid
\hbox{$u(\;\!\cdot\;\!-\gamma^*)=\tau(\gamma^*)u$
for all $\gamma^*\in\Gamma^*_\star$}\big\}
$$
with
$$
\h_\star^1
:=\big\{\varphi\in\H^1_{\rm loc}(\R,\C^2)\mid
\hbox{$T_\gamma\varphi=\varphi$ for all $\gamma\in\Gamma_\star$}\big\}.
$$

In the next proposition, we prove that the operator $\widehat{M_\star}$ is
analytically fibered in the sense of \cite[Def.~2.2]{GN98}. For this, we need to
introduce the Bloch variety
\begin{equation}\label{def_Bloch_variety}
\Sigma_\star:=\big\{(k,\lambda)\in Y^*_\star\times\R\mid
\lambda\in\sigma\big(\widehat{M_\star}(k)\big)\big\}.
\end{equation}

\begin{Proposition}[Fibering of the asymptotic Hamiltonians]\label{prop_fibered}
Let
$$
\widehat{M_\star}(\omega)\;\!\varphi
:=\left(w_\star\widehat D(0)+w_\star
\begin{pmatrix}
0 & \omega\\
\omega & 0
\end{pmatrix}\right)\varphi,
\quad\omega\in\C,~\varphi\in\h_\star^1.
$$
\begin{enumerate}
\item[(a)] The set
$$
\O_\star:=\big\{(\omega,z)\in\C\times\C\mid
z\in\rho\big(\widehat{M_\star}(\omega)\big)\big\}.
$$
is open in $\C\times\C$ and the map
$
\O_\star\ni(\omega,z)\mapsto\big(\widehat{M_\star}(\omega)-z\big)^{-1}\in\B(\h_\star)
$
is analytic in the variables $\omega$ and $z$.
\item[(b)] For each $\omega\in\C$, the operator $\widehat{M_\star}(\omega)$ has purely
discrete spectrum.
\item[(c)] If $\Sigma_\star$ is equipped with the topology induced by
$Y^*_\star\times\C$, then the projection $\pi_\R:\Sigma_\star\to\R$ given by
$\pi_\R(k,\lambda):=\lambda$ is proper.
\end{enumerate}
In particular, the operator $\widehat{M_\star}$ is analytically fibered in the sense
of \cite[Def.~2.2]{GN98}.
\end{Proposition}

\begin{proof}
(a) The operator $w_\star\widehat D(0)$ is self-adjoint on
$\h_\star^1\subset\h_\star$, and for each $\omega\in\C$ we have that
$
w_\star
\big(\begin{smallmatrix}
0 & \omega\\
\omega & 0
\end{smallmatrix}\big)
\in\B(\h_\star)
$.
Hence, for each $\omega\in\C$ the operator $\widehat{M_\star}(\omega)$ is closed in
$\h_\star$ and has domain $\h_\star^1$, and for each $x\in\R$ the operator
$\widehat{M_\star}(x)$ is self-adjoint on $\h_\star^1$. In particular, we infer by
functional calculus that
$$
\lim_{|t|\to\infty}
\big\|\big(\widehat{M_\star}(x)-it\big)^{-1}\big\|_{\B(\h_\star)}
\le\lim_{|t|\to\infty}\tfrac1{|t|}
=0\quad(t\in\R).
$$
Therefore, for each $y\in\R$ the set
$$
\Omega_y
:=\left\{it\in i\;\!\R\mid\left(\big\|\big(\widehat{M_\star}(x)
-it\big)^{-1}\big\|_{\B(\h_\star)}\right)^{-1}
>|y|\;\!\|w_\star\|_{\B(\h_\star)}\right\}
$$
is non-empty, and then the argument in \cite[Rem.~IV.3.2]{Kat95} guarantees that
$\Omega_y$ is contained in the resolvent set of $\widehat{M_\star}(x+iy)$. Thus, for
each $\omega\equiv x+iy\in\C$ the operator $\widehat{M_\star}(\omega)$ is closed in
$\h_\star$, has domain $\h_\star^1$, and non-empty resolvent set, and for each
$\varphi\in\h_\star^1$ the map
$\C\ni\omega\mapsto\widehat{M_\star}(\omega)\varphi\in\h_\star$ is linear and
therefore analytic. So, the collection
$\{\widehat{M_\star}(\omega)\}_{\omega\in\C}$ is an analytic family of type
(A) \cite[p.~16]{RS4}, and thus also an analytic family in the sense of Kato
\cite[p.~14]{RS4}. The claim is then a consequence of \cite[Thm.~XII.7]{RS4}.

(b) Since $\{\widehat{M_\star}(\omega)\}_{\omega\in\C}$ is an analytic family
of type (A), the operators $\widehat{M_\star}(\omega)$ have compact resolvent (and
thus purely discrete spectrum) either for all $\omega\in\C$ or for no $\omega\in\C$
\cite[Thm.~III.6.26 \& VII.2.4]{Kat95}. Therefore, to prove the claim, it is
sufficient to show that $\widehat{M_\star}(0)$ has compact resolvent. Now, we have
$$
\widehat{M_\star}(0)
=w_\star\widehat D(0)
=w_\star
\begin{pmatrix}
0 & -i\partial_\theta\\
-i\partial_\theta & 0
\end{pmatrix},
$$
where $-i\partial_\theta$ is a first order differential operator in $\ltwo(Y_\star)$
with periodic boundary conditions, and thus with purely discrete spectrum that
accumulates at $\pm\infty$. In consequence, each entry of the matrix operator
$$
\big(\widehat D(0)+i\big)^{-1}
=\begin{pmatrix}
-i\big(1+(i\partial_\theta)^2\big)^{-1}
& -i\partial_\theta\big(1+(i\partial_\theta)^2\big)^{-1}\\
-i\partial_\theta\big(1+(i\partial_\theta)^2\big)^{-1}&
-i\big(1+(i\partial_\theta)^2\big)^{-1}
\end{pmatrix}
$$
is compact in $\ltwo(Y_\star)$, so that $\big(\widehat D(0)+i\big)^{-1}$ is compact in
$\ltwo(Y_\star,\C^2)$. Since Lemma \ref{lemma_self}(a) implies that the norms on
$\ltwo(Y_\star,\C^2)$ and $\h_\star$ are equivalent, we infer that
$\big(\widehat D(0)+i\big)^{-1}$ is also compact in $\h_\star$. Finally, since
\begin{align}
\big(\widehat{M_\star}(0)+i\big)^{-1}
&=\big(\widehat{M_\star}(0)+iw_\star\big)^{-1}+\big(\widehat{M_\star}(0)+i\big)^{-1}
-\big(\widehat{M_\star}(0)+iw_\star\big)^{-1}\nonumber\\
&=\big(\widehat D(0)+i\big)^{-1}w_\star^{-1}-i\big(\widehat{M_\star}(0)+i\big)^{-1}
(1-w_\star)\big(\widehat D(0)+i\big)^{-1}w_\star^{-1},\label{eq_resolvent_bis}
\end{align}
with $w_\star^{-1}$ and $(1-w_\star)$ bounded in $\h_\star$, we obtain that
$\big(\widehat{M_\star}(0)+i\big)^{-1}$ is compact in $\h_\star$.

(c) Let $Y^*_\star\times\C$ be endowed with the topology induced by $\C\times\C$.
Point (a) implies that the set
$$
\Sigma^{\rm c}_\star
:=\big\{(k,z)\in Y^*_\star\times\C\mid z\in\rho\big(\widehat{M_\star}(k)\big)\big\}
$$
is open in $Y^*_\star\times\C$. Therefore, the set $\Sigma_\star$ is closed in
$Y^*_\star\times\C$ and the inclusion $\iota:\Sigma_\star\to Y_\star^*\times\C$ is a
closed map. Since the projection $\pi_\C:Y_\star^*\times\C\to\C$ given by
$\pi_\C(k,z):=z$ is also a closed map (because $Y^*_\star$ is compact, see
\cite[Ex.~7,~p.~171]{Mun00}) and $\pi_\R=\pi_\C\circ\iota$, we infer that $\pi_\R$ is
a closed map. Moreover, $\pi_\R$ is continuous because it is the restriction to the
subset $\Sigma_\star$ of the continuous projection $\pi_\C:Y_\star^*\times\C\to\C$. In
consequence, in order to prove that $\pi_\R$ is proper it is sufficient to show that
$\pi_\R^{-1}(\{\lambda\})$ is compact in $\Sigma_\star$ for each
$\lambda\in\pi_\R(\Sigma_\star)$. But since
$$
\pi_\R^{-1}(\{\lambda\})
=\big(\iota^{-1}\circ\pi_\C^{-1}\big)(\{\lambda\})
=\iota^{-1}(Y^*_\star\times\{\lambda\})
=(Y^*_\star\times\{\lambda\})\cap\Sigma_\star,
$$
this follows from compactness of $Y^*_\star$ and the closedness of $\Sigma_\star$ in
$Y^*_\star\times\C$.
\end{proof}

Proposition \ref{prop_fibered} can be combined with the theorem of Rellich
\cite[Thm.~VII.3.9]{Kat95} which, adapted to our notations, states\;\!:

\begin{Theorem}[Rellich]\label{thm_Rellich}
Let $\Omega\subset\C$ be a neighborhood of an interval $I_0\subset\R$ and let
$\{T(\omega)\}_{\omega\in\Omega}$ be a self-adjoint analytic family of type (A), with
each $T(\omega)$ having compact resolvent. Then there is a sequence of scalar-valued
functions $\lambda_n$ and a sequence of vector-valued functions $u_n$, all analytic on
$I_0$, such that for $\omega\in I_0$ the $\lambda_n(\omega)$ are the repeated
eigenvalues of $T(\omega)$ and the $u_n(\omega)$ form a complete orthonormal family of
the associated eigenvectors of $T(\omega)$.
\end{Theorem}

By applying this theorem to the family $\{\widehat{M_\star}(\omega)\}_{\omega\in\C}$,
we infer the existence of analytic eigenvalue functions
$\lambda_{\star,n}:Y^*_\star\to\R$ and analytic orthonormal eigenvector functions
$u_{\star,n}:Y^*_\star\to\h_\star$. We call band the graph
$\{\big(k,\lambda_{\star,n}(k)\big)\mid k\in Y^*_\star\}$ of the eigenvalue
function $\lambda_{\star,n}$, so that the Bloch variety $\Sigma_\star$ coincides with
the countable union of the bands (see \eqref{def_Bloch_variety}). Since the derivative
$\lambda'_{\star,n}$ of $\lambda_{\star,n}$ exists and is analytic, it is natural to
define the set of thresholds of the operator $M_\star$ as
\begin{equation}\label{eq_thresholds_star}
\Tau_\star
:=\bigcup_{n\in\N}\big\{\lambda\in\R\mid\hbox{$\exists\;\!k\in Y^*_\star$ such that
$\lambda=\lambda_{\star,n}(k)$ and $\lambda'_{\star,n}(k)=0$}\big\},
\end{equation}
and the set of thresholds of both $M_\ell$ and $M_{\rm r}$ as
\begin{equation}\label{eq_thresholds_M}
\Tau_M:=\Tau_\ell\cup\Tau_{\rm r}.
\end{equation}
Proposition \ref{prop_fibered}(b), together with the analyticity of the functions
$\lambda_{\star,n}$, implies that the set $\Tau_\star$ is discrete, with only possible
accumulation point at infinity. Furthermore, \cite[Thm.~XIII.85(e)]{RS4} implies that
the possible eigenvalues of $M_\star$ are contained in $\Tau_\star$. However, these
eigenvalues should be generated by locally (hence globally) flat bands, and one can
show their absence by adapting Thomas' argument \cite[Sec.~II]{Tho73} to our
setup\;\!:

\begin{Proposition}[Spectrum of the asymptotic Hamiltonians]\label{prop_spec_asymp}
The spectrum of $M_\star$ is purely absolutely continuous. In particular,
$$
\sigma(M_\star)=\sigma_{\rm ac}(M_\star)=\sigma_{\rm ess}(M_\star),
$$
with $\sigma_{\rm ac}(M_\star)$ the absolutely continuous spectrum of $M_\star$
and $\sigma_{\rm ess}(M_\star)$ the essential spectrum of $M_\star$.
\end{Proposition}

\begin{proof}
In view of \cite[Thm.~XIII.86]{RS4}, the claim follows once we prove the absence of
flat bands for $M_\star$. For this purpose, we use the version of the Thomas' argument
as presented in \cite[Sec.~1.3]{Sus00}. Accordingly, we first need to show that, for
$\omega=i\rho$ with $\rho\in\R$ large enough, the operator $\widehat{M_\star}(i\rho)$
is invertible and satisfies
\begin{equation}\label{eq_no_flat}
\lim_{|\rho|\to\infty}\big\|\widehat{M_\star}(i\rho)^{-1}\big\|_{\B(\h_\star)}=0.
\end{equation}
Let us start with the operator
$$
\widehat D(i\rho)=\begin{pmatrix}
0 & -i\partial_\theta+ i\rho\\
-i\partial_\theta+i\rho & 0
\end{pmatrix}
$$
acting on $\h_\star^1\subset\h_\star$. Since the family of functions
$\{e_n^\pm\}_{n\in\Z}$ given by
$$
e_n^+(\theta):=\tfrac1{\sqrt{p_\star}}\e^{2\pi in\theta/p_\star}
\begin{pmatrix}1\\0\end{pmatrix},
\quad e_n^-(\theta):=\tfrac1{\sqrt{p_\star}}\e^{2\pi in\theta/p_\star}
\begin{pmatrix}0\\1\end{pmatrix},
\quad\theta\in Y_\star,
$$
is an orthonormal basis of $\ltwo(Y_\star,\C^2)$, and since $\h_\star$ and
$\ltwo(Y_\star,\C^2)$ have equivalent norms, the family $\{e_n^\pm\}_{n\in\Z}$, with
extended variable $\theta\in\R$, is also a (non-orthogonal) basis for $\h_\star$, and
thus any $\varphi\in\h_\star^1$ can be expanded in $\h_\star$ as
$$
\varphi=\sum_{n\in\Z}\big(\widehat\varphi_n^+e_n^++\widehat\varphi_n^-e_n^-\big)
\quad\hbox{with}\quad
\widehat\varphi_n^\pm:=\langle\varphi,e_n^\pm\rangle_{\ltwo(Y_\star,\C^2)}.
$$
It follows that
\begin{align*}
\big\|\widehat D(\pm i\rho)\varphi\big\|^2_{\h_\star}
&=\left\|\;\!\sum_{n\in\Z}\left(\tfrac{2\pi n}{p_\star}\pm i\rho\right)
\big(\widehat\varphi_n^+e_n^-+\widehat\varphi_n^-e_n^+\big)\right\|^2_{\h_\star}\\
&\ge{\rm Const.}\left\|\sum_{n\in\Z}\left(\tfrac{2\pi n}{p_\star}\pm i\rho\right)
\big(\widehat\varphi_n^+e_n^-
+\widehat\varphi_n^-e_n^+\big)\right\|^2_{\ltwo(Y_\star,\C^2)}\\
&={\rm Const.}\sum_{n\in\Z}\left|\tfrac{2\pi n}{p_\star}\pm i\rho\right|^2
\left(|\widehat\varphi_n^+|^2+|\widehat\varphi_n^-|^2\right)\\
&\ge{\rm Const.}\;\!|\rho|^2\;\!\|\varphi\|^2_{\ltwo(Y_\star,\C^2)}\\
&\ge{\rm Const.}\;\!|\rho|^2\;\!\|\varphi\|^2_{\h_\star}.
\end{align*}
Thus, the operators $\widehat D(\pm i\rho)$ are injective with closed range and
satisfy in $\h_\star$ the relations
$$
\big(\Ran\widehat D(\pm i\rho)\big)^\bot
=\Ker\big(\widehat D(\pm i\rho)^*\big)
=\Ker\big(w_\star\widehat D(\mp i\rho)w_\star^{-1}\big)
=0.
$$
In consequence $\Ran\widehat D(\pm i\rho)=\h_\star$, and the operators
$\widehat D(\pm i\rho)$ are invertible with
$$
\big\|\widehat D(\pm i\rho)^{-1}\big\|_{\B(\h_\star)}\le{\rm Const.}\;\!|\rho|^{-1}.
$$
It follows that $\widehat{M_\star}(i\rho)$ is invertible too, with
$$
\big\|\widehat{M_\star}(i\rho)^{-1}\big\|_{\B(\h_\star)}
=\big\|\widehat D(i\rho)^{-1}w_\star^{-1}\big\|_{\B(\h_\star)}
\le{\rm Const.}\;\!|\rho|^{-1},
$$
which implies \eqref{eq_no_flat}.

Now, let us assume by contradiction that there exists $n\in\N$ such that
$\lambda_{\star,n}(k)$ is equal to a constant $c\in\R$ for all $k\in Y^*_\star$. Then
using the analyticity properties of $\widehat{M_\star}$ (Proposition
\ref{prop_fibered}) in conjunction with the analytic Fredholm alternative, one infers
that $c$ is an eigenvalue of $\widehat{M_\star}(\omega)$ for all $\omega\in\C$.
Letting $u(\omega)$ be the corresponding eigenfunction for
$\widehat{M_\star}(\omega)$, one obtains that
$\widehat{M_\star}(\omega)u(\omega)=cu(\omega)$ for all $\omega\in\C$. Choosing
$\omega=i\rho$ with $\rho\in\R$ and using the fact that $\widehat{M_\star}(i\rho)$ is
invertible, one thus obtains that
$$
u(i\rho)=c\;\!\widehat{M_\star}(i\rho)^{-1}u( i\rho)
\quad\hbox{with}\quad\|u(i\rho)\|_{\h_\star}=1,
$$
which contradicts \eqref{eq_no_flat}.
\end{proof}

\begin{Remark}
The absence of flat bands for the 3-dimensional Maxwell operator has been discussed
in \cite[Sec.~5]{Sus00}. However, the results of \cite{Sus00} do not cover the result
of Proposition \ref{prop_spec_asymp} since the weights considered in \cite{Sus00} are
block-diagonal and smooth while in Proposition \ref{prop_spec_asymp} the weights are
$\linfty$ positive-definite $2\times2$ matrices. Neither diagonality, nor smoothness is
assumed.
\end{Remark}

\section{Mourre theory and spectral results}\label{sec_conjugate}
\setcounter{equation}{0}

\subsection{Commutators}\label{sec_com}

In this section, we recall some definitions appearing in Mourre theory and provide a
precise meaning to the commutators mentioned in the introduction. We refer to
\cite{ABG96,Sah97} for more information and details.

Let $A$ be a self-adjoint operator in a Hilbert space $\H$ with domain $\dom(A)$, and
let $T\in\B(\H)$. For any $k\in\N$, we say that $T$ belongs to $C^k(A)$, with notation
$T\in C^k(A)$, if the map
\begin{equation}\label{eq_group}
\R\ni t\mapsto\e^{-itA}T\e^{itA}\in\B(\H)
\end{equation}
is strongly of class $C^k$. In the case $k=1$, one has $T\in C^1(A)$ if and only if
the quadratic form
$$
\dom(A)\ni\varphi\mapsto\langle\varphi,TA\;\!\varphi\rangle_\H
-\langle A\;\!\varphi,T\varphi\rangle_\H\in\C
$$
is continuous for the norm topology induced by $\H$ on $\dom(A)$. We denote by $[T,A]$
the bounded operator associated with the continuous extension of this form, or
equivalently $-i$ times the strong derivative of the function \eqref{eq_group} at
$t=0$.

If $H$ is a self-adjoint operator in $\H$ with domain $\dom(H)$ and spectrum
$\sigma(H)$, we say that $H$ is of class $C^k(A)$ if $(H-z)^{-1}\in C^k(A)$ for some
$z\in\C\setminus\sigma(H)$. In particular, $H$ is of class $C^1(A)$ if and only if the
quadratic form
$$
\dom(A)\ni\varphi\mapsto
\big\langle\varphi,(H-z)^{-1}A\;\!\varphi\big\rangle_\H
-\big\langle A\;\!\varphi,(H-z)^{-1}\varphi\big\rangle_\H\in\C
$$
extends continuously to a bounded form defined by the operator
$[(H-z)^{-1},A]\in\B(\H)$. In such a case, the set $\dom(H)\cap\dom(A)$ is a
core for $H$ and the quadratic form
$$
\dom(H)\cap\dom(A)\ni\varphi\mapsto\langle H\varphi,A\;\!\varphi\rangle_\H
-\langle A\;\!\varphi,H\varphi\rangle_\H\in\C
$$
is continuous in the natural topology of $\dom(H)$ (\emph{i.e.} the topology of the
graph-norm) \cite[Thm.~6.2.10(b)]{ABG96}. This form then extends uniquely to a
continuous quadratic form on $\dom(H)$ which can be identified with a continuous
operator $[H,A]$ from $\dom(H)$ to the adjoint space $\dom(H)^*$. In addition, one has
the identity
\begin{equation}\label{eq_resolvent}
[(H-z)^{-1},A]=-(H-z)^{-1}[H,A](H-z)^{-1},
\end{equation}
and the following result is verified \cite[Thm.~6.2.15]{ABG96}\;\!: If $H$ is of class
$C^k(A)$ for some $k\in\N$ and $\eta\in\S(\R)$ is a Schwartz function, then
$\eta(H)\in C^k(A)$.

A regularity condition slightly stronger than being of class $C^1(A)$ is defined as
follows\;\!: $H$ is of class $C^{1+\varepsilon}(A)$ for some $\varepsilon\in(0,1)$ if $H$
is of class $C^1(A)$ and if for some $z\in\C\setminus\sigma(H)$
$$
\big\|\e^{-itA}[(H-z)^{-1},A]\e^{itA}-[(H-z)^{-1},A]\big\|_{\B(\H)}
\le{\rm Const.}\;\!t^\varepsilon\quad\hbox{for all $t\in(0,1)$.}
$$
The condition $C^2(A)$ is stronger than $C^{1+\varepsilon}(A)$, which in turn is
stronger than $C^1(A)$.

We now recall the definition of two useful functions introduced in
\cite[Sec.~7.2]{ABG96}. For this, we need the following conventions\;\!: if
$E^H(\;\!\cdot\;\!)$ denotes the spectral projection-valued measure of $H$, then we set
$E^H(\lambda;\varepsilon):=E^H\big((\lambda-\varepsilon,\lambda+\varepsilon)\big)$ for
any $\lambda\in\R$ and $\varepsilon>0$, and if $S,T\in\B(\H)$, then we write
$S\approx T$ if $S-T$ is compact, and $S\lesssim T$ if there exists a compact operator
$K$ such that $S\le T+K$. With these conventions, we define for $H$ of class $C^1(A)$
the function $\varrho^A_H:\R\to(-\infty,\infty]$ by
$$
\varrho^A_H(\lambda)
:=\sup\big\{a\in\R\mid\hbox{$\exists\;\!\varepsilon>0$ such that
$a\;\!E^H(\lambda;\varepsilon)\le E^H(\lambda;\varepsilon)[iH,A]
E^H(\lambda;\varepsilon)$}\big\},
$$
and we define the function $\widetilde\varrho^A_H:\R\to(-\infty,\infty]$ by
$$
\widetilde\varrho^A_H(\lambda)
:=\sup\big\{a\in\R\mid\hbox{$\exists\;\!\varepsilon>0$ such that
$a\;\!E^H(\lambda;\varepsilon)\lesssim E^H(\lambda;\varepsilon)[iH,A]
E^H(\lambda;\varepsilon)$}\big\}.
$$
Note that the following equivalent definition of the function $\widetilde\varrho^A_H$
is often useful\;\!:
\begin{equation}\label{eq_alternate}
\widetilde\varrho^A_H(\lambda)
=\sup\big\{a\in\R\mid\hbox{$\exists\;\!\eta\in C^\infty_{\rm c}(\R,\R)$
such that $\eta(\lambda)\ne0$ and $a\;\!\eta(H)^2\lesssim\eta(H)[iH,A]\eta(H)$}\big\}.
\end{equation}
One says that $A$ is conjugate to $H$ at a point $\lambda\in\R$ if
$\widetilde\varrho^A_H(\lambda)>0$, and that $A$ is strictly conjugate to $H$ at
$\lambda$ if $\varrho^A_H(\lambda)>0$. It is shown in \cite[Prop.~7.2.6]{ABG96} that
the function $\widetilde\varrho^A_H:\R\to(-\infty,\infty]$ is lower semicontinuous,
that $\widetilde\varrho^A_H\ge\varrho^A_H$, and that
$\widetilde\varrho^A_H(\lambda)<\infty$ if and only if $\lambda\in\sigma_{\rm ess}(H)$.
In particular, the set of points where $A$ is
conjugate to $H$,
$$
\widetilde\mu^A(H):=\big\{\lambda\in\R\mid\widetilde\varrho^A_H(\lambda)>0\big\},
$$
is open in $\R$.

The main consequences of the existence of a conjugate operator $A$ for $H$ are given
in the theorem below, which is a particular case of \cite[Thm.~0.1~\&~0.2]{Sah97}. For
its statement, we use the notation $\sigma_{\rm p}(H)$ for the point spectrum of
$H$, and we recall that if $\G$ is an auxiliary Hilbert space, then an operator
$T\in\B(\H,\G)$ is locally $H$-smooth on an open set $I\subset\R$ if for each compact
set $I_0\subset I$ there exists $c_{I_0}\ge0$ such that
\begin{equation}\label{def_H_smooth}
\int_\R\d t\;\!\big\|T\;\!\e^{-itH}E^H(I_0)\varphi\big\|_\G^2
\le c_{I_0}\;\!\|\varphi\|_\H^2\quad\hbox{for each $\varphi\in\H$},
\end{equation}
and $T$ is (globally) $H$-smooth if \eqref{def_H_smooth} is satisfied with $E^H(I_0)$
replaced by the identity $1$.

\begin{Theorem}[Spectrum of $H$]\label{thm_spec_H}
Let $H,A$ be self-adjoint operators in a Hilbert space $\H$, let $\G$ be an
auxiliary Hilbert space, assume that $H$ is of class $C^{1+\varepsilon}(A)$ for some
$\varepsilon\in(0,1)$, and suppose there exist an open set $I\subset\R$, a number
$a>0$ and an operator $K\in\K(\H)$ such that
\begin{equation}\label{eq_Mourre}
E^H(I)\;\![iH,A]\;\!E^H(I)\ge a\;\!E^H(I)+K.
\end{equation}
Then
\begin{enumerate}
\item[(a)] each operator $T\in\B(\H,\G)$ which extends continuously to an element of
$\B\big(\dom(\langle A\rangle^s)^*,\G\big)$ for some $s>1/2$ is locally $H$-smooth on
$I\setminus\sigma_{\rm p}(H)$,
\item[(b)] $H$ has at most finitely many eigenvalues in $I$, each one of finite
multiplicity, and $H$ has no singular continuous spectrum in $I$.
\end{enumerate}
\end{Theorem}

\subsection{Conjugate operator for the free Hamiltonian}\label{sec_conj_free}

With the definitions of Section \ref{sec_fibering} at hand, we can construct a
conjugate operator for the operator $\widehat{M_\star}$. Our construction follows from
the one given in \cite[Sec.~3]{GN98}, but it is simpler because our base manifold
$Y^*_\star$ is one-dimensional. Indeed, thanks to Theorem \ref{thm_Rellich}, it is
sufficient to construct the conjugate operator band by band.

So, for each $n\in\N$, let
$\widehat\Pi_{\star,n}:=\{\widehat\Pi_{\star,n}(k)\}_{k\in\R}$ and
$\widehat\lambda'_{\star,n}:=\{\widehat\lambda'_{\star,n}(k)\}_{k\in\R}$ be
the bounded decomposable self-adjoint operators in $\H_{\tau,\star}$ defined by
$\tau$-equivariant continuation as in \eqref{eq_cov} and by the relations
$$
\widehat\Pi_{\star,n}(k)\varphi
:=\langle u_{\star,n}(k),\varphi\rangle_{\h_\star}u_{\star,n}(k)
\quad\hbox{and}\quad
\widehat\lambda'_{\star,n}(k)\varphi:=\lambda'_{\star,n}(k)\varphi,
\quad k\in Y^*_\star,~\varphi\in\h_\star.
$$
Set also $\Pi_{\star,n}:=\U^{-1}_\star\widehat\Pi_{\star,n}\U_\star$ and
$\widehat{Q_\star}:=\U_\star Q_\star\U_\star^{-1}$, with $Q_\star$ the operator of
multiplication by the variable in $\H_{w_\star}$
$$
(Q_\star\varphi)(x):=x\;\!\varphi(x),\quad\varphi\in\dom(Q_\star)
:=\big\{\varphi\in\H_{w_\star}\mid\|Q_\star\varphi\|_{\H_{w_\star}}<\infty\big\}.
$$

\begin{Remark}\label{rem_Q}
Since $Q_\star$ commutes with $w_\star^{-1}$, the operator $Q_\star$ is self-adjoint
in $\H_{w_\star}$ and essentially self-adjoint on $\S(\R,\C^2)\subset\H_{w_\star}$.
The definition and the domain of $Q_\star$ are independent of the specific weight
$w_\star^{-1}$ appearing in the scalar product of $\H_{w_\star}$. The insistence on
the label $\star=\ell,{\rm r}$ is only motivated by a notational need that will result
helpful in the next sections.
\end{Remark}

For any compact interval $I\subset\R\setminus\Tau_\star$, we define the finite set
$\N(I):=\big\{n\in\N\mid\lambda_{\star,n}^{-1}(I)\ne\varnothing\big\}$. Finally, we set
\begin{equation}\label{def_D_star}
\D_\star:=\U_\star\S(\R,\C^2)\subset\big\{u\in C^\infty(\R,\h_\star)\mid
\hbox{$u(\;\!\cdot\;\!-\gamma^*)=\tau(\gamma^*)u$ for all
$\gamma^*\in\Gamma^*_\star$}\big\}.
\end{equation}
Then we can define the symmetric operator $\widehat A_{\star,I}$ in $\H_{\tau,\star}$
by
\begin{equation}\label{def_A_star}
\widehat A_{\star,I}u
:=\tfrac12\sum_{n\in\N(I)}\widehat\Pi_{\star,n}
\big(\widehat\lambda'_{\star,n}\;\!\widehat{Q_\star}
+\widehat{Q_\star}\;\!\widehat\lambda'_{\star,n}\big)\widehat\Pi_{\star,n}u,
\quad u\in\D_\star.
\end{equation}

\begin{Theorem}[Mourre estimate for $\widehat{M_\star}$]\label{thm_Mourre_Free}
Let $I\subset\R\setminus\Tau_\star$ be a compact interval. Then
\begin{enumerate}
\item[(a)] the operator $\widehat A_{\star,I}$ is essentially self-adjoint on
$\D_\star$ and on any other core for $\widehat{Q_\star^2}$, with closure denoted by
the same symbol,
\item[(b)] the operator $\widehat{M_\star}$ is of class $C^2(\widehat A_{\star,I})$,
\item[(c)] there exists $c_I>0$ such that
$\varrho_{\widehat{M_\star}}^{\widehat A_{\star,I}}\ge c_I$.
\end{enumerate}
\end{Theorem}

\begin{proof}
(a) The claim is a consequence of Nelson's criterion of self-adjointness
\cite[Thm.~X.37]{RS2} applied to the triple $(\widehat A_{\star,I},N_\star,\D_\star)$,
where $N_\star:=\widehat{Q_\star^2}+1$ and
$\widehat{Q_\star^2}:=\U_\star Q_\star^2\U_\star^{-1}$. Indeed, the operator $N_\star$
is essentially self-adjoint on $\D_\star=\U_\star\S(\R,\C^2)$ since $Q_\star^2$ is
essentially self-adjoint on $\S(\R,\C^2)$. In addition, since $\widehat A_{\star,I}$
is composed of the bounded operators $\widehat\Pi_{\star,n}$ and
$\widehat\lambda'_{\star,n}$ which are analytic in the variable $k\in Y_\star^*$ and
$\widehat{Q_\star}$ acts as $i\partial_k$ in $\H_{\tau,\star}$, a direct computation
gives
$$
\big\|\widehat A_{\star,I}u\big\|_{\H_{\tau,\star}}
\le{\rm Const.}\;\!\|N_\star u\|_{\H_{\tau,\star}},\quad u\in\D_\star.
$$
Similarly, a direct computation using the boundedness and the analyticity of
$\widehat\Pi_{\star,n}$ and $\widehat\lambda'_{\star,n}$ implies that
$$
\big|\langle\widehat A_{\star,I}u,N_\star u\rangle_{\H_{\tau,\star}}
-\langle N_\star u,\widehat A_{\star,I}u\rangle_{\H_{\tau,\star}}\big|
\le{\rm Const.}\;\!\langle N_\star u,u\rangle_{\H_{\tau,\star}},\quad u\in\D_\star.
$$
In both inequalities, we used the fact that
$\dom(\widehat{Q_\star^2})\subset\dom(\widehat{Q_\star})$. As a consequence,
$\widehat A_{\star,I}$ is essentially self-adjoint on $\D_\star$ and on any other core
for $N_\star$.

(b) The set
$$
\E_\star:=\big\{ u\in C^\infty(\R,\h_\star)\mid\hbox{$u(\;\!\cdot\;\!-\gamma^*)
=\tau(\gamma^*)u$ for all $\gamma^*\in\Gamma^*_\star$}\big\}\supset\D_\star
$$
is a core for $N_\star$. So, it follows from point (a) that $\widehat A_{\star,I}$ is
essentially self-adjoint on $\E_\star$. Moreover, since $\widehat{M_\star}(k)$ is
analytic in $k\in\R$ and satisfies the covariance relation \eqref{eq_cov}, we obtain
that $(\widehat{M_\star}-z)\E_\star\subset\E_\star$ for any
$z\in\C\setminus\sigma(\widehat{M_\star})$. Since the same argument applies to the
resolvent, we obtain that $(\widehat{M_\star}-z)^{-1}\E_\star=\E_\star$. 
Therefore, we have the inclusion $(\widehat{M_\star}-z)^{-1}u\in\dom(\widehat A_{\star,I})$ 
for each $u\in\E_\star$, and a calculation using \eqref{eq_resolvent} gives
\begin{align*}
\big\langle u,[i(\widehat{M_\star}-z)^{-1},\widehat A_{\star,I}]u
\big\rangle_{\H_{\tau,\star}}
&=\big\langle u,-(\widehat{M_\star}-z)^{-1}
[i\widehat{M_\star},\widehat A_{\star,I}]
(\widehat{M_\star}-z)^{-1}u\big\rangle_{\H_{\tau,\star}}\\
&=\bigg\langle u,-(\widehat{M_\star}-z)^{-1}\sum_{n\in\N(I)}
\widehat\Pi_{\star,n}\;\!\big|\widehat\lambda'_{\star,n}\big|^2\;\!
\widehat\Pi_{\star,n}(\widehat{M_\star}-z)^{-1}u
\bigg\rangle_{\H_{\tau,\star}}.
\end{align*}
Since
$
\sum_{n\in\N(I)}\widehat\Pi_{\star,n}\;\!\big|\widehat\lambda'_{\star,n}\big|^2
\;\!\widehat\Pi_{\star,n}\in\B(\H_{\tau,\star})
$,
it follows that $\widehat{M_\star}$ is of class $C^1(\widehat A_{\star,I})$ with
\begin{equation}\label{eq_first_com}
[i\widehat{M_\star},\widehat A_{\star,I}]
=\sum_{n\in\N(I)}\widehat\Pi_{\star,n}\;\!\big|\widehat\lambda'_{\star,n}\big|^2
\;\!\widehat\Pi_{\star,n}.
\end{equation}
Finally, since $\widehat\Pi_{\star,n}\in C^1(\widehat A_{\star,I})$ and
$\widehat\lambda'_{\star,n}\in C^1(\widehat A_{\star,I})$ for each $n\in\N$, we infer
from \eqref{eq_first_com} and \cite[Prop.~5.1.5]{ABG96} that $\widehat{M_\star}$ is of
class $C^2(\widehat A_{\star,I})$.

(c) Using point (b) and the definition of the operators $\widehat\Pi_{\star,n}$, we
obtain for all $\eta\in C^\infty_{\rm c}(I,\R)$ and $k\in Y^*_\star$ that
\begin{align*}
&\eta\big(\widehat{M_\star}(k)\big)
[i\widehat{M_\star},\widehat A_{\star,I}](k)\;\!
\eta\big(\widehat{M_\star}(k)\big)\\
&=\eta\big(\widehat{M_\star}(k)\big)\Bigg(\sum_{n\in\N(I)}\widehat\Pi_{\star,n}(k)
\big|\widehat\lambda'_{\star,n}(k)\big|^2\;\!\widehat\Pi_{\star,n}(k)\Bigg)
\eta\big(\widehat{M_\star}(k)\big)\\
&\ge c_I\;\!\eta\big(\widehat{M_\star}(k)\big)\Bigg(\sum_{n\in\N(I)}
\widehat\Pi_{\star,n}(k)^2\Bigg)\eta\big(\widehat{M_\star}(k)\big)\\
&=c_I\;\!\eta\big(\widehat{M_\star}(k)\big)^2
\end{align*}
with
$
c_I:=\min_{n\in\N(I)}\min_{\{k\in Y^*_\star\mid\lambda_{\star,n}(k)\in I\}}
\big|\lambda'_n(k)\big|^2
$.
Thus, by using the definition of the scalar product in $\H_{\tau,\star}$, we infer
that
$$
\eta\big(\widehat{M_\star}\big)[i\widehat{M_\star},\widehat A_{\star,I}]
\eta\big(\widehat{M_\star}\big)\\
\ge c_I\;\!\eta\big(\widehat{M_\star}\big)^2,
$$
which, together with the definition \eqref{eq_alternate}, implies the claim.
\end{proof}

Since the operator $\widehat A_{\star,I}$ is essentially self-adjoint on
$\D_\star=\U_\star\S(\R,\C^2)$ and on any other core for $\widehat{Q_\star^2}$,
it follows by Theorem \ref{thm_Mourre_Free}(a) that the inverse Bloch-Floquet
transform of $\widehat A_{\star,I}$,
$$
A_{\star,I}:=\U^{-1}_\star\widehat A_{\star,I}\U_\star,
$$
is essentially self-adjoint on $\S(\R,\C^2)$ and on any other core for $Q_\star^2$.
Therefore, the results (b) and (c) of Theorem \ref{thm_Mourre_Free} can be restated as
follows\;\!: the operator $M_\star$ is of class $C^2(A_{\star,I})$ and there exists
$c_I>0$ such that $\varrho_{M_\star}^{A_{\star,I}}\ge c_I$. Combining these results
for $\star=\ell$ and $\star={\rm r}$, one obtains a conjugate operator for the free
Hamiltonian $M_0=M_\ell\oplus M_{\rm r}$ introduced in Section \ref{sec_free}. Namely,
for any compact interval $I\subset\R\setminus\Tau_M$, the operator
$$
A_{0,I}:=A_{\ell,I}\oplus A_{{\rm r},I}
$$
satisfies the following properties\;\!:
\begin{enumerate}
\item[(a')] the operator $A_{0,I}$ is essentially self-adjoint on
$\S(\R,\C^2)\oplus\S(\R,\C^2)$ and on any set $\E\oplus\E$ with $\E$ a core for
$Q_\star^2$, with closure denoted by the same symbol,
\item[(b')] the operator $M_0$ is of class $C^2(A_{0,I})$,
\item[(c')] there exists $c_I>0$ such that $\varrho_{M_0}^{A_{0,I}}\ge c_I$.
\end{enumerate}

\begin{Remark}\label{rem_ess_M_0}
What precedes implies in particular that the free Hamiltonian $M_0$ has purely
absolutely continuous spectrum except at the points of $\Tau_M$, where it may have
eigenvalues. However, we already know from Proposition \ref{prop_spec_asymp} that this
does not occur. Therefore,
$$
\sigma(M_0)
=\sigma_{\rm ac}(M_0)
=\sigma_{\rm ess} (M_0)
=\sigma_{\rm ess}(M_\ell)\cup\sigma_{\rm ess}(M_{\rm r}).
$$
\end{Remark}

\subsection{Conjugate operator for the full Hamiltonian}\label{sec_conj_full}

In this section, we show that the operator $JA_{0,I}J^*$ is a conjugate operator for
the full Hamiltonian $M$ introduced in Section \ref{sec_full}. We start with the proof
of the essential self-adjointness of $JA_{0,I}J^*$ in $\H_w$. We use the notation $Q$
(see Remark \ref{rem_Q}) for the operator of multiplication by the variable in $\H_w$,
$$
(Q\varphi)(x):=x\;\!\varphi(x),\quad\varphi\in\dom(Q)
:=\big\{\varphi\in\H_w\mid\|Q\varphi\|_{\H_w}<\infty\big\}.
$$

\begin{Proposition}\label{prop_A_esa}
For each compact interval $I\subset\R\setminus\Tau_M$ the operator $A_I:=JA_{0,I}J^*$
is essentially self-adjoint on $\S(\R,\C^2)$ and on any other core for $Q^2$, with
closure denoted by the same symbol.
\end{Proposition}

\begin{proof}
First, we observe that since $J^*\S(\R,\C^2)\subset\dom(Q_0^2)$ with
$Q_0:=Q_\ell\oplus Q_{\rm r}$ the operator $A_I$ is well-defined and symmetric on
$\S(\R,\C^2)\subset\H_w$ due to point (a') above. Next, to prove the claim, we use
Nelson's criterion of essential self-adjointness \cite[Thm.~X.37]{RS2} applied to the
triple $\big(A_I,N,\S(\R,\C^2)\big)$ with $N:=Q^2+1$.

For this, we note that $\S(\R,\C^2)$ is a core for $N$ and that the operators
$\frac{Q_\star}{Q^2+1}$, $j_\star$, $w_\star w^{-1}j_\star$, $\Pi_{\star,n}$ and
$\U^{-1}_\star\widehat\lambda'_{\star,n}\U_\star$ are bounded in $\H_w$. Moreover, we
verify with direct calculations on $\S(\R,\C^2)$ that the operators $\Pi_{\star,n}$
and $\U^{-1}_\star\widehat\lambda'_{\star,n}\U_\star$ belong to $C^1(Q_\star)$ (in
$\H_{w_\star}$), and that their commutators $[\Pi_{\star,n},Q_\star]$ and
$[\U^{-1}_\star\widehat\lambda'_{\star,n}\U_\star,Q_\star]$ belong to $C^1(Q)$ (in
$\H_w$). Then a short computation using these properties gives the bound
$$
\|A_I\varphi\|_{\H_w}\le{\rm Const.}\;\!\|N\varphi\|_{\H_w},\quad\varphi\in\S(\R,\C^2),
$$
and a slightly longer computation using the same properties shows that
$$
\big|\langle A_I\varphi,N\varphi\rangle_{\H_w}
-\langle N\varphi,A_{I}\varphi\rangle_{\H_w}\big|
\le{\rm Const.}\;\!\langle N\varphi,\varphi\rangle_{\H_w},\quad\varphi\in\S(\R,\C^2).
$$
Thus, the hypotheses of Nelson's criterion are satisfied, and the claim follows.
\end{proof}

In order to prove that $A_I$ is a conjugate operator for $M$, we need two preliminary
lemmas. They involve the two-Hilbert spaces difference of resolvents of $M_0$ and $M:$
$$
B(z):=J\left(M_0-z\right)^{-1}-(M-z)^{-1}J,\quad z\in\C\setminus\R.
$$

\begin{Lemma}\label{lemma_compact_B}
For each $z\in\C\setminus\R$, one has the inclusion $B(z)\in\K(\H_0,\H_w)$.
\end{Lemma}

\begin{proof}
One has for $(\varphi_\ell,\varphi_{\rm r})\in\H_0$
\begin{align}
B(z)(\varphi_\ell,\varphi_{\rm r})
&=\sum_{\star\in\{\ell,{\rm r}\}}
\big(j_\star(M_\star-z)^{-1}-(M-z)^{-1}j_\star\big)\varphi_\star\nonumber\\
&=\sum_{\star\in\{\ell,{\rm r}\}}
\big\{\big((M_\star-z)^{-1}-(M-z)^{-1}\big)j_\star \varphi_\star
+[j_\star,(M_\star-z)^{-1}]\varphi_\star\big\}.\label{eq_two_terms}
\end{align}
Thus, an application of the standard result \cite[Thm.~4.1]{Sim05} taking into account
the properties of $j_\star$ implies that the operator
$[j_\star,(M_\star-z)^{-1}]$ is compact. This proves the claim for the second
term in \eqref{eq_two_terms}.

For the first term in \eqref{eq_two_terms}, we have the equalities
\begin{align}
&\big((M_\star-z)^{-1}-(M-z)^{-1}\big)\;\!j_\star\nonumber\\
&=(M-z)^{-1}(M-M_\star)\;\!j_\star (M_\star-z)^{-1}
+(M-z)^{-1}(M-M_\star)[(M_\star-z)^{-1},j_\star]\nonumber\\
&=(M-z)^{-1}j_\star(w-w_\star)D(M_\star-z)^{-1}
+(M-z)^{-1}(w-w_\star)[D,j_\star](M_\star-z)^{-1}\nonumber\\
&\quad+(M-z)^{-1}(M-M_\star)[(M_\star-z)^{-1},j_\star],\label{eq_pre_trick}
\end{align}
with $j_\star(w-w_\star)$ and $[D,j_\star]$ matrix-valued functions vanishing at
$\pm\infty$. Thus, the operator in the first term in \eqref{eq_two_terms} is also
compact, which concludes the proof.
\end{proof}

\begin{Lemma}\label{lema_compact_closure}
For each $z\in\C\setminus\R$ and each compact interval $I\subset\R\setminus\Tau_M$,
one has the inclusion
$$
\overline{B(z)A_{0,I}\upharpoonright\dom(A_{0,I})}\in\K(\H_0,\H_w).
$$
\end{Lemma}

\begin{proof}
Since $A_{0,I}$ is essentially self-adjoint on $\S(\R,\C^2)\oplus\S(\R,\C^2)$, it is
sufficient to show that
$$
\overline{B(z)A_{0,I}\upharpoonright\big(\S(\R,\C^2)\oplus\S(\R,\C^2)\big)}
\in\K(\H_0,\H_w).
$$
Furthermore, we have $A_{0,I}=A_{\ell,I}\oplus A_{{\rm r},I}$, and each operator
$A_{\star,I}$ acts on $\S(\R,\C^2)$ as a sum $Q_\star F_{\star,I}+G_{\star,I}$, with
$F_{\star,I},G_{\star,I}$ bounded operators in $\H_{w_\star}$ mapping the set
$\S(\R,\C^2)$ into $\dom(Q_\star)$. These facts, together with the compactness result
of Lemma \ref{lemma_compact_B} and \eqref{eq_two_terms}, imply that it is sufficient
to show that
$$
\overline{\big((M_\star-z)^{-1}-(M-z)^{-1}\big)j_\star
\;\!Q_\star\upharpoonright\S(\R,\C^2)}\in\K(\H_{w_\star},\H_w).
$$
and
$$
\overline{[j_\star,(M_\star-z)^{-1}]Q_\star
\upharpoonright\S(\R,\C^2)}\in\K(\H_{w_\star},\H_w).
$$
Now, if one takes Assumption \ref{ass_weight} into account, the proof of these
inclusions is similar to the proof of Lemma \ref{lemma_compact_B}. We leave the
details to the reader.
\end{proof}

Next, we will need the following theorem which is a direct consequence of Theorem 3.1
and Corollaries 3.7-3.8 of \cite{RT13_2}.

\begin{Theorem}\label{thm_2_Hilbert}
Let $H_0,A_0$ be self-adjoint operators in a Hilbert space $\H_0$, let $H$ be a
self-adjoint operator in a Hilbert space $\H$, let $J\in\B(\H_0,\H)$, and let
$$
{\cal B}(z):=J(H_0-z)^{-1}-(H-z)^{-1}J,\quad z\in\C\setminus\R.
$$
Suppose there exists a set $\D\subset\dom(A_0J^*)$ such that $JA_0J^*$ is essentially
self-adjoint on $\D$, with $A$ its self-adjoint extension. Finally, assume that
\begin{enumerate}
\item[(i)] $H_0$ is of class $C^1(A_0)$,
\item[(ii)] for each $z\in\C\setminus\R$, one has ${\cal B}(z)\in\K(\H_0,\H)$,
\item[(iii)] for each $z\in\C\setminus\R$, one has
$\overline{{\cal B}(z)A_0\upharpoonright\dom(A_0)}\in\K(\H_0,\H)$,
\item[(iv)] for each $\eta\in C_{\rm c}^\infty(\R)$, one has
$\eta(H)\left(JJ^*-1\right)\eta(H)\in\K(\H)$.
\end{enumerate}
Then $H$ is of class $C^1(A)$ and
$\widetilde\varrho_H^A\ge\widetilde\varrho_{H_0}^{A_0}$. In particular, if $A_0$ is
conjugate to $H_0$ at $\lambda\in\R$, then $A$ is conjugate to $H$ at $\lambda$.
\end{Theorem}

We are now ready to prove a Mourre estimate for $M:$

\begin{Theorem}[Mourre estimate for $M$]\label{thm_regul}
Let $I\subset\R\setminus\Tau_M$ be a compact interval. Then $M$ is of class
$C^1(A_I)$, and
$$
\widetilde\varrho_M^{A_I}(\lambda)
\ge\widetilde\varrho_{M_0}^{A_{0,I}}(\lambda)
=\min\big\{\widetilde\varrho_{M_\ell}^{A_{\ell,I}}(\lambda),
\widetilde\varrho_{M_{\rm r}}^{A_{{\rm r},I}}(\lambda)\big\}
>0\quad\hbox{for every $\lambda\in I$.}
$$
\end{Theorem}

\begin{proof}
Theorem \ref{thm_Mourre_Free} and its restatement at the end of Section
\ref{sec_conj_free} give us the estimate
$$
\min\big\{\widetilde\varrho_{M_\ell}^{A_{\ell,I}}(\lambda),
\widetilde\varrho_{M_{\rm r}}^{A_{{\rm r},I}}(\lambda)\big\}>0
\quad\hbox{for every $\lambda\in I$.}
$$
In addition, the equality
$
\widetilde\varrho_{M_0}^{A_{0,I}}
=\min\big\{\widetilde\varrho_{M_\ell}^{A_{\ell,I}},
\widetilde\varrho_{M_{\rm r}}^{A_{{\rm r},I}}\big\}
$
is a consequence of the definition of $A_{0,I}$ as a direct sum of $A_{\ell,I}$ and
$A_{{\rm r},I}$ (see \cite[Prop. 8.3.5]{ABG96}).

So, it only remains to show the inequality
$\widetilde\varrho_M^{A_I}\ge\widetilde\varrho_{M_0}^{A_{0,I}}$ to prove the claim.
For this, we apply Theorem \ref{thm_2_Hilbert} with $H_0=M_0$, $H=M$ and
$A_0=A_{0,I}$, starting with the verification of its assumptions (i)-(iv)\;\!: the
assumptions (i), (ii) and (iii) of Theorem \ref{thm_2_Hilbert} follow from point (b')
above, Lemma \ref{lemma_compact_B}, and Lemma \ref{lema_compact_closure},
respectively. Furthermore, the assumption (iv) of Theorem \ref{thm_2_Hilbert} follows
from the fact that for any $\eta\in C^\infty_{\rm c}(\R)$ we have the inclusion
$$
\eta(M)\left(JJ^*-1\right)\eta(M)
=\eta(M)\big(w_\ell w^{-1}j^2_\ell+w_{\rm r}w^{-1}j^2_{\rm r}-1\big)(Q)\;\!\eta(M)
\in\K(\H_w),
$$
since
$$
w_\ell w^{-1}j^2_\ell+w_{\rm r}w^{-1}j^2_{\rm r}-1 
=(w_\ell-w)\;\!j_\ell^2w^{-1}+(w_{\rm r}-w)\;\!j_{\rm r}^2w^{-1}
+(j_\ell^2+j_{\rm r}^2-1)
$$
is a matrix-valued function vanishing at $\pm\infty$. These facts, together with
Proposition \ref{prop_A_esa} and the inclusion $\S(\R,\C^2)\subset\dom(A_0 J^*)$,
imply that all the assumptions of Theorem \ref{thm_2_Hilbert} are satisfied. We thus
obtain that $\widetilde\varrho_M^{A_I}\ge\widetilde\varrho_{M_0}^{A_{0,I}}$, as
desired.
\end{proof}

\subsection{Spectral properties of the full Hamiltonian}\label{sec_spec}

In this section, we determine the spectral properties of the full Hamiltonian $M$. We
start by proving that $M$ has the same essential spectrum as the free Hamiltonian
$M_0:$

\begin{Proposition}\label{proposition_ess_M}
One has
$$
\sigma_{\rm ess}(M)
=\sigma_{\rm ess}(M_0)
=\sigma(M_\ell)\cup\sigma(M_{\rm r}).
$$
\end{Proposition}

To prove Proposition \ref{proposition_ess_M}, we first need two preliminary lemmas. In
the first lemma, we use the notation $\chi_\Lambda$ for the characteristic function of a
Borel set $\Lambda\subset\R$.

\begin{Lemma}\label{lemma_ess_M}
\begin{enumerate}
\item[(a)] The operator $M$ is locally compact in $\H_w$, that is,
$\chi_B(Q)(M-i)^{-1}\in\K(\H_w)$ for each bounded Borel set $B\subset\R$.
\item[(b)] Let $\zeta\in C^\infty_{\rm c}\big(\R,[0,\infty)\big)$ satisfy $\zeta(x)=1$
for $|x|\le1$ and $\zeta(x)=0$ for $|x|\ge2$, and set $\zeta_n(x):=\zeta(x/n)$ for all
$x\in\R$ and $n\in\N\setminus\{0\}$. Then
$$
\lim_{n\to \infty}\big\|[M,\zeta_n(Q)](M-i)^{-1}\big\|_{\B(\H_w)}=0.
$$
\end{enumerate}
\noindent
Moreover, the results of (a) and (b) also hold true for the operators $M_0$ and $Q_0$
in the Hilbert space $\H_0$.
\end{Lemma}

\begin{proof}
(a) A direct computation shows that
$$
\chi_B(Q)(D-i)^{-1}
=\begin{pmatrix}
i\;\!\chi_B(Q)(1+P^2)^{-1} & \chi_B(Q)P\;\!(1+P^2)^{-1}\\
\chi_B(Q)P\;\!(1+P^2)^{-1} & i\;\!\chi_B(Q)(1+P^2)^{-1}
\end{pmatrix},
$$
which implies that $\chi_B(Q)(D-i)^{-1}$ is compact in $\ltwo(\R,\C^2)$ since every
entry of the matrix is compact in $\ltwo(\R)$ (see \cite[Thm.~4.1]{Sim05}). Given that
$\ltwo(\R,\C^2)$ and $\H_w$ have equivalent norms by Lemma \ref{lemma_self}(b), it
follows that $\chi_B(Q)(D-i)^{-1}$ is also compact in $\H_w$. Finally, the resolvent
identity (similar to \eqref{eq_resolvent_bis})
$$
(M-i)^{-1}
=(D-i)^{-1}w^{-1}+i\;\!(D-i)^{-1}(w^{-1}-1)(M-i)^{-1}
$$
shows that $\chi_B(Q)(M-i)^{-1}$ is the sum of two compact operators in $\H_w$, and
hence compact in $\H_w$. The same argument also shows that the operators $M_\star$ are
locally compact in $\H_{w_\star}$, and thus that $M_0=M_\ell\oplus M_{\rm r}$ is
locally compact in $\H_0=\H_{w_\ell}\oplus\H_{w_{\rm r}}$.

(b) Let $\varphi\in\dom(M)=\H^1(\R,\C^2)$. Then a direct computation taking into
account the inclusion $\zeta_n(Q)\varphi\in\dom(M)$ gives
$$
[M,\zeta_n(Q)]\varphi
=w
\begin{pmatrix}
0 & [P,\zeta_n(Q)]\\
[P,\zeta_n(Q)] & 0
\end{pmatrix}
\varphi
=-\tfrac inw\;\!\zeta'_n(Q)
\begin{pmatrix}
0 & 1\\
1 & 0
\end{pmatrix}
\varphi.
$$
In consequence, we obtain that
$\big\|[M,\zeta_n(Q)](M-i)^{-1}\big\|_{\B(\H_w)}\le{\rm Const.}\;\!n^{-1}$
which proves the claim. As before, the same argument also applies to the operators
$M_\star$ in $\H_{w_\star}$, and thus to the operator $M_0=M_\ell\oplus M_{\rm r}$ in
$\H_0=\H_{w_\ell}\oplus\H_{w_{\rm r}}$.
\end{proof}

Lemma \ref{lemma_ess_M} is needed to prove that the essential spectra of $M$ and $M_0$
can be characterised in terms of \emph{Zhislin sequences} (see
\cite[Def.~10.4]{HS96}). Zhislin sequences are particular types of Weyl sequences
supported at infinity as in the following lemma\;\!:

\begin{Lemma}[Zhislin sequences]\label{lemma_ess_M_bis}
Let $\lambda\in\R$. Then $\lambda\in\sigma_{\rm ess}(M)$ if and only if there exists
a sequence $\{\phi_m\}_{m\in\N\setminus\{0\}}\subset\dom(M)$, called Zhislin sequence, such that\;\!:
\begin{enumerate}
\item[(a)] $\|\phi_m\|_{\H_w}=1$ for all $m\in\N\setminus\{0\}$,
\item[(b)] for each $m\in\N\setminus\{0\}$, one has $\phi_m(x)=0$ if $|x|\le m$,
\item[(c)] $\lim_{m\to\infty}\|(M-\lambda)\phi_m\|_{\H_w}=0$.
\end{enumerate}
Similarly, $\lambda\in\sigma_{\rm ess}(M_0)$ if and only if there exists a sequence
$\{\phi^0_m\}_{m\in\N\setminus\{0\}}\subset\dom(M_0)$ which meets the properties (a), (b), (c)
relative to the operator $M_0$.
\end{Lemma}

\begin{proof}
In view of Lemma \ref{lemma_ess_M}, the claim can be proved by repeating step by step
the arguments in the proof of \cite[Thm.~10.6]{HS96}.
\end{proof}

We are now ready to complete the description of the essential spectrum of $M:$

\begin{proof}[Proof of Proposition \ref{proposition_ess_M}]
Take $\lambda\in\sigma_{\rm ess}(M)$, let
$\{\phi_m\}_{m\in\N\setminus\{0\}}\subset\dom(M)$ be an associated Zhislin sequence,
and define for each $m\in\N\setminus\{0\}$
$$
\phi^0_m:=c_m\;\!(j_\ell\;\!\phi_m,j_{\rm r}\;\!\phi_m)\in\dom(M_0)
\quad\hbox{with}\quad
c_m:=\big\|\big(j_\ell\;\!\phi_m,j_{\rm r}\;\!\phi_m\big)\big\|_{\H_0}^{-1}.
$$
Then one has $\big\|\phi^0_m\big\|_{\H_0}=1$ for all $m\in\N\setminus\{0\}$ and
$\phi^0_m(x)=0$ if $|x|\le m$. Furthermore, using successively the facts that
$j_\ell\;\!j_{\rm r}=0$, that $(j_\ell+j_{\rm r})\phi_m=\phi_m$, and that
$1=\|\phi_m\|^2_{\H_w}=\langle\phi_m,w^{-1}\phi_m\rangle_{\ltwo(\R,\C^2)}$, one
obtains that
$$
c^{-2}_m
=\big\langle\phi_m,\big(w^{-1}_\ell j_\ell
+w^{-1}_{\rm r}j_{\rm r}\big)\phi_m\big\rangle_{\ltwo(\R,\C^2)}
=1+\big\langle\phi_m,\big(w^{-1}_\ell j_\ell
+w^{-1}_{\rm r}j_{\rm r}-w^{-1}\big)\phi_m\big\rangle_{\ltwo(\R,\C^2)},
$$
which implies that $\lim_{m\to \infty}c_m=1$ due to Assumption \ref{ass_weight}.

Now, consider the inequalities
\begin{align*}
\big\|(M_0-\lambda)\phi^0_m\big\|_{\H_0}^2
&=c^2_m\sum_{\star\in\{\ell,\rm r\}}\big\|(M_\star-\lambda)
\;\!j_\star\phi_m\big\|^2_{\H_{w_\star}}\\
&\le c^2_m\sum_{\star\in\{\ell,\rm r\}}
\left(\big\|(M-\lambda)\;\!j_\star\phi_m\big\|_{\H_{w_\star}}
+\big\|(M_\star-M)\;\!j_\star\phi_m\big\|_{\H_{w_\star}}\right)^2\\
&\le c^2_m\sum_{\star\in\{\ell,\rm r\}}
\left(\big\|j_\star(M-\lambda)\phi_m\big\|_{\H_{w_\star}}
+\big\|[M,j_\star]\;\!\phi_m\big\|_{\H_{w_\star}}
+\big\|(w_\star-w)D\;\!j_\star\phi_m\big\|_{\H_{w_\star}}\right)^2.
\end{align*}
From the property (c) of Zhislin sequences, the boundedness of $j_\star$, and the
equivalence of the norms of $\H_{w_\star}$ and $\H_w$, one gets that
$$
\lim_{m\to\infty}\big\|\;\!j_\star(M-\lambda)\phi_m\big\|_{\H_{w_\star}}
\le{\rm Const.}\lim_{m\to\infty}\big\|(M-\lambda)\phi_m\big\|_{\H_w}=0.
$$
Moreover, one has
$
[M,j_\star]\;\!\phi_m
=-iwj_\star'\big(\begin{smallmatrix}0&1\\1&0\end{smallmatrix}\big)\phi_m
$,
with $j_\star'$ supported in $[-1,1]$. This implies that $[M,j_\star]\;\!\phi_m=0$.
For the same reason, one has $Dj_\star\phi_m=j_\star D\phi_m$, with the latter vector
supported in $x\le-m$ if $\star=\ell$ and in $x\ge m$ if $\star={\rm r}$. This,
together with Assumption \ref{ass_weight}, implies that
$$
\big\|(w_\star-w)Dj_\star\phi_m\big\|_{\H_{w_\star}}
\le{\rm Const.}\;\!\langle m\rangle^{-1-\varepsilon}\|D\phi_m\|_{\H_{w_\star}}
\le{\rm Const.}\;\!\langle m\rangle^{-1-\varepsilon}\|M\phi_m\|_{\H_w}.
$$
The last inequality, along with the equality
$\lim_{m\to\infty}\|M\phi_m\|_{\H_w}=|\lambda|$ (which follows from the property (c)
of Zhislin sequences), gives
$$
\lim_{m\to\infty}\big\|(w_\star-w)Dj_\star\phi_m\big\|_{\H_{w_\star}}=0.
$$
Putting all the pieces together, we obtain that
$\lim_{m\to\infty}\big\|(M_0-\lambda)\phi^0_m\big\|_{\H_0}=0$. This concludes the
proof that $\{\phi^0_m\}_{m\in\N\setminus\{0\}}$ is a Zhislin sequence for the
operator $M_0$ and the point $\lambda\in\sigma_{\rm ess}(M)$, and thus that
$\sigma_{\rm ess}(M)\subset\sigma_{\rm ess}(M_0)$.

For the opposite inclusion, take $\{\phi^0_m\}_{m\in\N\setminus\{0\}}\subset\dom(M_0)$
a Zhislin sequence for the operator $M_0$ and the point
$\lambda\in\sigma_{\rm ess}(M_0)$, and assume that
$\lambda\in\sigma_{\rm ess}(M_{\rm r})$ (if
$\lambda\notin\sigma_{\rm ess}(M_{\rm r})$, then $\lambda\in\sigma_{\rm ess}(M_\ell)$
and the same proof applies if one replaces ``right'' with ``left''). By extracting the
nonzero right components from $\phi^0_m$ and normalising, we can form a new Zhislin
sequence $\{(0,\phi^{\rm r}_m)\}_{m\in\N\setminus\{0\}}$ for $M_0$ with
$\{\phi^{\rm r}_m\}_{m\in\N\setminus\{0\}}\subset\dom(M_{\rm r})$ a Zhislin sequence
for $M_{\rm r}$. Then we can construct as follows a new Zhislin sequence for
$M_{\rm r}$ with vectors supported in $[m,\infty)$\;\!: Let
$\zeta^{\rm r}\in C^\infty_{\rm c}(\R,[0,1])$ satisfy $\zeta^{\rm r}(x)=0$ for $x\le1$
and $\zeta^{\rm r}(x)=1$ for $x\ge2$, set $\zeta^{\rm r}_m(x):=\zeta^{\rm r}(x/m)$ for
all $x\in\R$ and $m\in\N\setminus\{0\}$, and choose $n_m^{\rm r}\in\N$ such that
$
\big\|\chi_{[-n_m^{\rm r},\infty)}(Q_{\rm r})\phi_m^{\rm r}\big\|_{\H_{w_{\rm r}}}
\in(1-1/m,1]
$.
Next, define for each $m\in\N\setminus\{0\}$
$$
\widetilde\phi_m^{\rm r}
:=d_m\;\!\zeta^{\rm r}_m(Q_{\rm r})T_{k_m^{\rm r}p_{\rm r}}\phi_m^{\rm r}
\in\dom(M_{\rm r}),
$$
with
$
d_m:=\big\|\zeta^{\rm r}_m(Q_{\rm r})T_{k_m^{\rm r}p_{\rm r}}
\phi_m^{\rm r}\big\|^{-1}_{\H_{w_{\rm r}}}
$,
$k_m^{\rm r}\in\N$ such that $k_m^{\rm r}p_{\rm r}\ge n_m^{\rm r}+2m$, and
$T_{k_m^{\rm r}p_{\rm r}}$ the operator of translation by $k_m^{\rm r}p_{\rm r}$. One
verifies easily that $\lim_{m\to\infty}d_m=1$ and that $\widetilde\phi_m^{\rm r}$ has
support in $[m,\infty)$. Furthermore, since the operators $M_{\rm r}$ and
$T_{k_m^{\rm r}p_{\rm r}}$ commute, one also has
\begin{align*}
\big\|(M_{\rm r}-\lambda)\widetilde\phi_m^{\rm r}\big\|_{\H_{w_{\rm r}}}
&\le d_m\;\!\big\|\zeta^{\rm r}_m(Q_{\rm r})T_{k_m^{\rm r}p_{\rm r}} 
(M_{\rm r}-\lambda) \phi_m^r\big\|_{\H_{w_{\rm r}}}
+d_m\;\!\big\|[M_{\rm r},\zeta^{\rm r}_m(Q_{\rm r})]T_{k_m^{\rm r}p_{\rm r}}
\phi_m^{\rm r}\big\|_{\H_{w_{\rm r}}} \\
&\le d_m\;\!\big\|(M_{\rm r}-\lambda)\phi_m^r\big\|_{\H_{w_{\rm r}}}
+{\rm Const.}\;\!m^{-1}\big\|\phi_m^{\rm r}\big\|_{\H_{w_{\rm r}}},
\end{align*}
which implies that
$
\lim_{m\to\infty}
\big\|(M_{\rm r}-\lambda)\widetilde\phi_m^{\rm r}\big\|_{\H_{w_{\rm r}}}=0
$.
Thus, $\{\widetilde\phi_m^{\rm r}\}_{m\in\N\setminus\{0\}}\subset\dom(M_{\rm r})$ is a
new Zhislin sequence for $M_{\rm r}$ with $\widetilde\phi_m^{\rm r}$ supported in
$[m,\infty)$.

Now, define for each $m\in\N\setminus\{0\}$
$$
\phi_m:=b_m\;\!\widetilde\phi_m^{\rm r}\in\dom(M)
\quad\hbox{with}\quad
b_m:=\big\|\widetilde\phi_m^{\rm r}\big\|^{-1}_{\H_w}.
$$
Then one has $\|\phi_m\|_{\H_w}=1$ for all $m\in\N\setminus\{0\}$ and $\phi_m(x)=0$
if $x\le m$. Furthermore, using that
$$
1
=\;\!\big\|\widetilde\phi_m^{\rm r}\big\|^2_{\H_{w_{\rm r}}}
=\langle\widetilde\phi_m^{\rm r},
w^{-1}_{\rm r}\widetilde\phi_m^{\rm r}\rangle_{\ltwo(\R,\C^2)},
$$
one obtains that
$$
b_m^{-2}
:=\big\langle\widetilde\phi_m^{\rm r},w^{-1}\widetilde\phi_m^{\rm r}
\big\rangle_{\ltwo(\R,\C^2)}
=1+\big\langle\widetilde\phi_m^{\rm r},
\big(w^{-1}-w^{-1}_{\rm r}\big)\widetilde\phi_m^{\rm r}\big\rangle_{\ltwo(\R,\C^2)},
$$
which implies that $\lim_{m\to \infty}b_m=1$ due to Assumption \ref{ass_weight}. Now,
consider the inequality
$$
\big\|(M-\lambda)\phi_m\big\|_{\H_w}
=b_m\;\!\big\|(M-\lambda)\widetilde\phi_m^{\rm r}\big\|_{\H_w}
\le b_m\;\!\big\|(M_{\rm r}-\lambda)\;\!\widetilde\phi_m^{\rm r}\big\|_{\H_w}
+b_m\;\!\big\|(w-w_{\rm r})D\widetilde\phi_m^{\rm r}\big\|_{\H_w}.
$$
From the property (c) of Zhislin sequences and the equivalence of the norms of
$\H_{w_{\rm r}}$ and $\H_w$, one gets that
$$
\lim_{m\to\infty}\big\|(M_{\rm r}-\lambda)\;\!\widetilde\phi_m^{\rm r}\big\|_{\H_w}
\le{\rm Const.}\lim_{m\to\infty}
\big\|(M_{\rm r}-\lambda)\;\!\widetilde\phi_m^{\rm r}\big\|_{\H_{w_{\rm r}}}=0.
$$
Moreover, since $D\widetilde\phi_m^{\rm r}$ is supported in $[m,\infty)$, one infers
again from Assumption \ref{ass_weight} that
$$
\big\|(w-w_{\rm r})D\widetilde\phi_m^{\rm r}\big\|_{\H_w}
\le{\rm Const.}\;\!\langle m\rangle^{-1-\varepsilon}
\big\|D\widetilde\phi_m^{\rm r}\big\|_{\H_w}
\le{\rm Const.}\;\!\langle m\rangle^{-1-\varepsilon}
\big\|M_{\rm r}\widetilde\phi_m^{\rm r}\big\|_{\H_{w_{\rm r}}}.
$$
The last inequality, along with the equality
$
\lim_{m\to\infty}\big\|M_{\rm r}\widetilde\phi_m^{\rm r}\big\|_{\H_{w_{\rm r}}}
=|\lambda|
$,
gives
$$
\lim_{m\to\infty}\big\|(w-w_{\rm r})D\widetilde\phi_m^{\rm r}\big\|_{\H_w}=0.
$$
Putting all the pieces together, we obtain that
$\lim_{m\to\infty}\big\|(M-\lambda)\phi_m\big\|_{\H_w}=0$. This concludes the proof
that $\{\phi_m\}_{m\in\N\setminus\{0\}}$ is a Zhislin sequence for the operator $M$
and the point $\lambda\in\sigma_{\rm ess}(M_0)$, and thus that
$\sigma_{\rm ess}(M_0)\subset\sigma_{\rm ess}(M)$. In consequence, we obtained that
$\sigma_{\rm ess}(M)=\sigma_{\rm ess}(M_0)$, which completes the proof in view of
Remark \ref{rem_ess_M_0}.
\end{proof}

In order to determine more precise spectral properties of $M$, we now prove that for
each compact interval $I\subset\R\setminus\Tau_M$ the Hamiltonian $M$ is of class
$C^{1+\varepsilon}(A_I)$ for some $\varepsilon\in(0,1)$, which is a regularity
condition slightly stronger than the condition $M$ of class $C^1(A_I)$ already
established in Theorem \ref{thm_regul}. We start by giving a convenient formula for
the commutator $[(M-z)^{-1},A_I]$, $z\in\C\setminus\R$, in the form sense on
$\S(\R,\C^2):$
\begin{align*}
&[(M-z)^{-1},A_I]\\
&=\big((M-z)^{-1}J-J(M_0-z)^{-1}\big)A_{0,I}J^*
-JA_{0,I}\big(J^*(M-z)^{-1}-(M_0-z)^{-1}J^*\big)\\
&\quad+J\;\![(M_0-z)^{-1},A_{0,I}]J^*\\
&=\sum_{\star\in\{\ell,{\rm r}\}}\big\{
\big((M-z)^{-1}-(M_\star-z)^{-1}\big)j_\star A_{\star,I}Z^\star
+[(M_\star-z)^{-1},j_\star]A_{\star,I}Z^\star\\
&\quad-j_\star A_{\star,I}Z^\star\big((M-z)^{-1}-(M_\star-z)^{-1}\big)
+j_\star A_{\star,I}[(M_\star-z)^{-1},Z^\star]
+j_\star[(M_\star-z)^{-1},A_{\star,I}]Z^\star\big\}\\
&=\sum_{\star\in\{\ell,{\rm r}\}}
\big(C_\star+j_\star[(M_\star-z)^{-1},A_{\star,I}]Z^\star\big)
\end{align*}
with
$$
Z^\star
:=w_\star w^{-1}j_\star
=j_\star-(w-w_\star)\;\!j_\star w^{-1},
$$
and
\begin{align*}
C_\star
&:=\big((M-z)^{-1}-(M_\star-z)^{-1}\big)\;\!j_\star A_{\star,I}Z^\star
+[(M_\star-z)^{-1},j_\star]A_{\star,I}Z^\star\\
&\quad-j_\star A_{\star,I}Z^\star\big((M-z)^{-1}-(M_\star-z)^{-1}\big)
+j_\star A_{\star,I}[(M_\star-z)^{-1},Z^\star].
\end{align*}
As already shown in the previous section, all the terms in $C_\star$ extend to bounded
operators, and we keep the same notation for these extensions.

In order to show that $(M-z)^{-1}\in C^{1+\varepsilon}(A_I)$, it is enough to prove
that $j_\star[(M_\star-z)^{-1},A_{\star,I}]Z^\star\in C^1(A_I)$ and to check
that
\begin{equation}\label{eq_regularity}
\big\|\e^{-itA_I}C_\star\e^{itA_I}-C_\star\big\|_{\B(\H_w)}
\le{\rm Const.}\;\!t^\varepsilon\quad\hbox{for all $t\in(0,1)$.}
\end{equation}
Since the first proof reduces to computations similar to the ones presented in the
previous section, we shall concentrate on the proof of \eqref{eq_regularity}. First of
all, algebraic manipulations as presented in \cite[pp.~325-326]{ABG96} or
\cite[Sec.~4.3]{RST_1} show that for all $t\in(0,1)$
\begin{align*}
\big\|\e^{-itA_I}C_\star\e^{itA_I}-C_\star\big\|_{\B(\H_w)}
&\le{\rm Const.}\;\!\Big(\big\|\sin(tA_I)C_\star\big\|_{\B(\H_w)}
+\big\|\sin(tA_I)(C_\star)^*\big\|_{\B(\H_w)}\Big)\\
&\le{\rm Const.}\;\!\Big(\big\|tA_I\;\!(tA_I+i)^{-1}C_\star\big\|_{\B(\H_w)}
+\big\|tA_I\;\!(tA_I+i)^{-1}(C_\star)^*\big\|_{\B(\H_w)}\Big).
\end{align*}
Furthermore, if we set $A_t:=tA_I\;\!(tA_I+i)^{-1}$ and
$\Lambda_t:=t\;\!\langle Q\rangle(t\;\!\langle Q\rangle+i)^{-1}$, we obtain that
$$
A_t=\big(A_t+i(tA_I +i)^{-1}A_I\;\!\langle Q\rangle^{-1}\big)\Lambda_t
$$
with $A_I\langle Q\rangle^{-1}\in\B(\H_w)$ due to \eqref{def_A_star}. Thus, since
$\big\|A_t+i(tA_I +i)^{-1}A_I\;\!\langle Q\rangle^{-1}\big\|_{\B(\H_w)}$ is bounded by
a constant independent of $t\in(0,1)$, it is sufficient to prove that
$$
\|\Lambda_t C_\star\|_{\B(\H_w)}+\|\Lambda_t(C_\star)^*\|_{\B(\H_w)}
\le{\rm Const.}\;\!t^\varepsilon\quad\hbox{for all $t\in(0,1)$,}
$$
and to prove this estimate it is sufficient to show that the operators
$\langle Q\rangle^\varepsilon C_\star$ and $\langle Q\rangle^\varepsilon(C_\star)^*$
defined in the form sense on $\S(\R,\C^2)$ extend continuously to elements of
$\B(\H_w)$. Finally, some lengthy but straightforward computations show that these two
last conditions are implied by the following two lemmas\;\!:

\begin{Lemma}\label{lemma_M_Q}
$M_\star$ is of class $C^1(\langle Q\rangle^\alpha)$ for each
$\star\in\{\ell,{\rm r}\}$ and $\alpha\in [0,1]$.
\end{Lemma}

\begin{proof}
One can verify directly that the unitary group generated by the operator
$\langle Q\rangle^\alpha$ leaves the domain $\dom(M_\star)=\H^1(\R,\C^2)$ invariant
and that the commutator $[M_\star,\langle Q\rangle^\alpha]$ defined in the form sense
on $\S(\R,\C^2)$ extends continuously to a bounded operator. Since the set
$\S(\R,\C^2)$ is a core for $M_\star$, these properties together with
\cite[Thm.~6.3.4(a)]{ABG96} imply the claim.
\end{proof}

\begin{Lemma}\label{lemma_j_star_C1}
One has $j_\star\in C^1(A_{\star,I})$ for each $\star\in\{\ell,{\rm r}\}$ and each
compact interval $I\subset\R\setminus\Tau_M$.
\end{Lemma}

\begin{proof}
By using the commutator expansions \cite[Thm.~5.5.3]{ABG96} and \eqref{def_A_star},
one gets the following equalities in form sense on $\S(\R,\C^2):$
\begin{align*}
&\big[j_\star,A_{\star,I}\big]\\
&=\tfrac1{\sqrt{2\pi}}\int_0^1\d\tau\int_\R\d x\,\e^{i\tau x Q_\star}
[Q_\star,A_{\star,I}]\e^{i(1-\tau)xQ_\star}\widehat{j_\star'}(x)\\
&=\tfrac1{2\sqrt{2\pi}}\sum_{n\in\N(I)}\int_0^1\d\tau\int_\R\d x\,\e^{i\tau x Q_\star}
\big[Q_\star,\Pi_{\star,n}\big(\widecheck\lambda'_{\star,n}\;\!Q_\star
+Q_\star\;\!\widecheck\lambda'_{\star,n}\big)\Pi_{\star,n}\big]
\e^{i(1-\tau)xQ_\star}\widehat{j_\star'}(x)\\
&=\tfrac1{2\sqrt{2\pi}}\sum_{n\in\N(I)}\int_0^1\d\tau\int_\R\d x\,\e^{i\tau x Q_\star}
\big\{[Q_\star,\Pi_{\star,n}\widecheck\lambda'_{\star,n}\Pi_{\star,n}]Q_\star
+Q_\star[Q_\star,\Pi_{\star,n}\widecheck\lambda'_{\star,n}\Pi_{\star,n}]\\
&\quad+\big[Q_\star,\Pi_{\star,n}\widecheck\lambda_{\star,n}[Q_\star,\Pi_{\star,n}]\big]
-\big[Q_\star,[Q_\star,\Pi_{\star,n}]\;\!\widecheck\lambda_{\star,n}
\Pi_{\star,n}\big]\big\}\e^{i(1-\tau)xQ_\star}\widehat{j_\star'}(x)
\end{align*}
with $\widecheck\lambda'_{\star,n}:=\U_\star^{-1}\widehat\lambda'_{\star,n}\U_\star$
and with each commutator in the last equality extending continuously to a bounded
operator. Since $\widehat{j_\star'}$ is integrable, the last two terms give bounded
contributions. Furthermore, the first two terms can be rewritten as
$$
f\big([Q_\star,\Pi_{\star,n}\widecheck\lambda'_{\star,n}\Pi_{\star,n}]\big)
Q_\star+Q_\star f\big(\big[Q_\star,\Pi_{\star,n}\widecheck\lambda'_{\star,n}
\Pi_{\star,n}\big]\big)
$$
with
$$
f:\B(\H_{w_\star})\to\B(\H_{w_\star}),
~~B\mapsto f(B):=\tfrac1{2\sqrt{2\pi}}\sum_{n\in\N(I)}\int_0^1\d\tau\int_\R\d x\,
\e^{i\tau x Q_\star}B\e^{i(1-\tau)xQ_\star}\widehat{j_\star'}(x).
$$
But, since
$
[Q_\star,\Pi_{\star,n}\widecheck\lambda'_{\star,n}\Pi_{\star,n}]
\in C^k(Q_\star)
$
for each $k\in\N$, and since $\widehat{j_\star'}$ is a Schwartz function, one infers
from \cite[Thm.~5.5.3]{ABG96} that the operator
$f\big([Q_\star,\Pi_{\star,n}\widecheck\lambda'_{\star,n}\Pi_{\star,n}]\big)$
is regularising in the Besov scale associated to the operator $Q_\star$. This implies
in particular that the operators
$
f\big([Q_\star,\Pi_{\star,n}\widecheck\lambda'_{\star,n}\Pi_{\star,n}]\big)
Q_\star
$
and
$
Q_\star f\big([Q_\star,\Pi_{\star,n}\widecheck\lambda'_{\star,n}\Pi_{\star,n}]\big)
$
extend continuously to bounded operators, as desired.
\end{proof}

We can now give in the next theorem a description of the structure of the spectrum of
the full Hamiltonian $M$. The next theorem also shows that the set $\Tau_M$ can be
interpreted as the set of thresholds in the spectrum of $M:$

\newpage

\begin{Theorem}\label{thm_spec_M}
In any compact interval $I\subset\R\setminus\Tau_M$, the operator $M$ has at most
finitely many eigenvalues, each one of finite multiplicity, and no singular continuous
spectrum.
\end{Theorem}

\begin{proof}
The computations at the beginning of this section together with Lemmas \ref{lemma_M_Q}
\& \ref{lemma_j_star_C1} imply that $M$ is of class $C^{1+\varepsilon}(A_I)$ for some
$\varepsilon\in(0,1)$, and Theorem \ref{thm_regul} implies that the condition
\eqref{eq_Mourre} of Theorem~\ref{thm_spec_H} is satisfied on $I$. So, one can apply
Theorem \ref{thm_spec_H}(b) to conclude.
\end{proof}

\begin{Remark}
As a final remark about the spectral properties of the operator $M$, let us mention
that the techniques used in this work do not provide any further information about the
existence of eigenvalues at thresholds or embedded in the continuous spectrum. For
additive perturbations, powerful techniques have been developed over the last decades,
and these methods can be applied to several Schrödinger-type operators. On the other
hand, for multiplicative perturbations, abstract methods have not been developed so
far. To the best of our knowledge, there exist only a few results about embedded
eigenvalues in the case of quantum walks with multiplicative perturbations as for
example in \cite{KMS,MSW,MS19}, or more abstractly in \cite{B13}. For photonic
crystals, especially in the presence of bi-anisotropic media, eigenvalues at
thresholds or embedded in the continuous spectrum certainly deserve an independent
study and this could be the subject of future investigations.
\end{Remark}

\section{Scattering theory}\label{sec_scattering}
\setcounter{equation}{0}

\subsection{Scattering theory in a two-Hilbert spaces setting}\label{sec_scatt_two}

We discuss in this section the existence and the completeness, under smooth
perturbations, of the local wave operators for self-adjoint operators in a two-Hilbert
spaces setting. Namely, given two self-adjoint operators $H_0,H$ in Hilbert spaces
$\H_0,\H$ with spectral measures $E^{H_0},E^H$, an identification operator
$J\in\B(\H_0,\H)$, and an open set $I\subset\R$, we recall criteria for the existence
and the completeness of the strong limits
$$
W_\pm(H,H_0,J,I):=\slim_{t\to\pm\infty}\e^{itH}J\e^{-itH_0}E^{H_0}(I)
$$
under the assumption that the two-Hilbert spaces difference of resolvents
$$
J(H_0-z)^{-1}-(H-z)^{-1}J,\quad z\in\C\setminus\R,
$$
factorises as a product of a locally $H$-smooth operator on $I$ and a locally
$H_0$-smooth operator on $I$.

We start by recalling some facts related to the notion of $J$-completeness. Let
$\NN_\pm(H,J,I)$ be the subsets of $\H$ defined by
$$
\NN_\pm(H,J,I)
:=\left\{\varphi\in\H\mid\lim_{t\to\pm\infty}
\big\|J^*\e^{-itH}E^H(I)\varphi\big\|_{\H_0}=0\right\}.
$$
Then it is clear that $\NN_\pm(H,J,I)$ are closed subspaces of $\H$ and that
$E^H(\R\setminus I)\;\!\H\subset\NN_\pm(H,J,I)$, and it is shown in
\cite[Sec.~3.2]{Yaf92} that $H$ is reduced by $\NN_\pm(H,J,I)$ and that
$$
\overline{\Ran\big(W_\pm(H,H_0,J,I)\big)}\perp\NN_\pm(H,J,I).
$$
In particular, one has the inclusion
$$
\overline{\Ran\big(W_\pm(H,H_0,J,I)\big)}\subset E^H(I)\H\ominus\NN_\pm(H,J,I),
$$
which motivates the following definition\;\!:

\begin{Definition}[$J$-completeness]\label{def_J_complete}
Assume that the local wave operators $W_\pm(H,H_0,J,I)$ exist. Then
$W_\pm(H,H_0,J,I)$ are $J$-complete on $I$ if
$$
\overline{\Ran\big(W_\pm(H,H_0,J,I)\big)}=E^H(I)\H\ominus\NN_\pm(H,J,I).
$$
\end{Definition}

\begin{Remark}
In the particular case $\H_0=\H$ and $J=1_\H$, the $J$-completeness on $I$ reduces to
the completeness of the local wave operators $W_\pm(H,H_0,J,I)$ on $I$ in the usual
sense. Namely, $\overline{\Ran\big(W_\pm(H,H_0,1_\H,I)\big)}=E^H(I)\H$, and the
operators $W_\pm(H,H_0,1_\H,I)$ are unitary from $E^{H_0}(I)\H$ to $E^H(I)\H$.
\end{Remark}

The following criterion for $J$-completeness has been established in
\cite[Thm.~3.2.4]{Yaf92}\;\!:

\begin{Lemma}
If the local wave operators $W_\pm(H,H_0,J,I)$ and $W_\pm(H_0,H,J^*,I)$ exist, then
$W_\pm(H,H_0,J,I)$ are $J$-complete on $I$.
\end{Lemma}

For the next theorem, we recall that the spectral support $\supp_H(\varphi)$ of a
vector $\varphi\in\H$ with respect to $H$ is the smallest closed set $I\subset\R$ such
that $E^H(I)\varphi=\varphi$.

\begin{Theorem}\label{thm_wave}
Let $H_0,H$ be self-adjoint operators in Hilbert spaces $\H_0,\H$ with spectral
measures $E^{H_0},E^H$, $J\in\B(\H_0,\H)$, $I\subset\R$ an open set, and $\G$ an
auxiliary Hilbert space. For each $z\in\C\setminus\R$, assume there exist
$T_0(z)\in\B(\H_0,\G)$ locally $H_0$-smooth on $I$ and $T(z)\in\B(\H,\G)$ locally
$H$-smooth $I$ such that
$$
J(H_0-z)^{-1}-(H-z)^{-1}J=T(z)^*T_0(z).
$$
Then the local wave operators
\begin{equation}\label{eq_wave}
W_\pm(H,H_0,J,I)=\slim_{t\to\pm\infty}\e^{itH}J\e^{-itH_0}E^{H_0}(I)
\end{equation}
exist, are $J$-complete on $I$, and satisfy the relations
$$
W_\pm(H,H_0,J,I)^*=W_\pm(H_0,H,J^*,I)
\quad\hbox{and}\quad
W_\pm(H,H_0,J,I)\;\!\eta(H_0)=\eta(H)\;\!W_\pm(H,H_0,J,I)
$$
for each bounded Borel function $\eta:\R\to\C$.
\end{Theorem}

\begin{proof}
We adapt the proof of \cite[Thm.~7.1.4]{ABG96} to the two-Hilbert spaces setting. The
existence of the limits \eqref{eq_wave} is a direct consequence of the following
claims\;\!: for each $\varphi_0\in\H_0$ with $I_0:=\supp_{H_0}(\varphi_0)\subset I$
compact, and for each $\eta\in C_{\rm c}^\infty(I)$ such that $\eta\equiv1$ on a
neighbourhood of $I_0$, we have that
\begin{equation}\label{eq_equiv_wave}
\slim_{t\to\pm\infty}\eta(H)\e^{itH}J\e^{-itH_0}\varphi_0~\hbox{exist}
\quad\hbox{and}\quad
\lim_{t\to\pm\infty}\big\|\big(1-\eta(H)\big)\e^{itH}J\e^{-itH_0}\varphi_0\big\|_\H=0.
\end{equation}
To prove the first claim in \eqref{eq_equiv_wave}, take $\varphi\in\H$ and $t\in\R$,
and observe that the operators $W(t):=\eta(H)\e^{itH}J\e^{-itH_0}$ satisfy for
$z\in\C\setminus\R$ and $s\le t$
\begin{align*}
&\big|\big\langle(H-\overline z)^{-1}\varphi,
\big(W(t)-W(s)\big)(H_0-z)^{-1}\varphi_0\big\rangle_\H\big|\\
&=\left|\int_s^t\d u\,\tfrac\d{\d u}\,\big\langle\e^{-iuH}\overline\eta(H)\varphi,
(H-z)^{-1}J(H_0-z)^{-1}\e^{-iuH_0}\varphi_0\big\rangle_\H\right|\\
&=\left|\int_s^t\d u\,\big\langle\e^{-iuH}\overline\eta(H)\varphi,
(H-z)^{-1}(HJ-JH_0)(H_0-z)^{-1}\e^{-iuH_0}\varphi_0\big\rangle_\H\right|\\
&=\left|\int_s^t\d u\,\big\langle\e^{-iuH}\overline\eta(H)\varphi,
\big(J(H_0-z)^{-1}-(H-z)^{-1}J\big)\e^{-iuH_0}\varphi_0\big\rangle_\H\right|\\
&=\left|\int_s^t\d u\,\big\langle T(z)\e^{-iuH}\overline\eta(H)\varphi,
T_0(z)\e^{-iuH_0}\varphi_0\big\rangle_\G\right|\\
&\le\left(\int_s^t\d u\,
\big\|T(z)\e^{-iuH}\overline\eta(H)\varphi\big\|^2_\G\right)^{1/2}
\left(\int_s^t\d u\,\big\|T_0(z)\e^{-iuH_0}\varphi_0\big\|^2_\G\right)^{1/2}\\
&\le c_{I_1}^{1/2}\;\!\|\varphi\|_\H
\left(\int_s^t\d u\,\big\|T_0(z)\e^{-iuH_0}\varphi_0\big\|^2_\G\right)^{1/2},
\end{align*}
with $I_1:=\supp(\eta)$ and $c_{I_1}$ the constant appearing in the definition
\eqref{def_H_smooth} of a locally $H$-smooth operator. Since the set
$(H-\overline z)^{-1}\H$ is dense in $\H$ and $T_0$ is locally $H_0$-smooth on $I$, it
follows that $\big\|\big(W(t)-W(s)\big)(H_0-z)^{-1}\varphi_0\big\|_\H\to0$ as
$s\to\infty$ or $t\to-\infty$. Applying this result with $\varphi_0$ replaced by
$(H_0-z)\varphi_0$ we infer that
$$
\big\|\big(W(t)-W(s)\big)\varphi_0\big\|_\H
=\big\|\big(W(t)-W(s)\big)(H_0-z)^{-1}(H_0-z)\varphi_0\big\|_\H\to0
$$
as $s\to\infty$ or $t\to-\infty$, which proves the first claim in
\eqref{eq_equiv_wave}.

To prove the second claim in \eqref{eq_equiv_wave}, we take
$\eta_0\in C_{\rm c}^\infty(I)$ such that $\eta_0\equiv1$ on $I_0$ and
$\eta\;\!\eta_0=\eta_0$. Then we have $\varphi_0=\eta_0(H_0)\varphi_0$ and
$$
\big(1-\eta(H)\big)J\;\!\eta_0(H_0)
=\big(1-\eta(H)\big)\big(J\;\!\eta_0(H_0)-\eta_0(H)J\big),
$$
and thus the second claim in \eqref{eq_equiv_wave} follows from
$$
\lim_{t\to\pm\infty}
\big\|\big(J\;\!\eta_0(H_0)-\eta_0(H)J\big)\e^{-itH_0}\varphi_0\big\|_\H=0.
$$
Since the vector space generated by the functions $\R\ni x\mapsto(x-z)^{-1}\in\C$,
$z\in\C\setminus\R$, is dense in $C_0(\R)$, it is sufficient to show that
$$
\lim_{t\to\pm\infty}
\big\|\big(J(H_0-z)^{-1}-(H-z)^{-1}J\big)\e^{-itH_0}\varphi_0\big\|_\H=0,
\quad z\in\C\setminus\R.
$$
Now, we have for every $\varphi\in\H$
\begin{align*}
\big|\big\langle\varphi,
\big(J(H_0-z)^{-1}-(H-z)^{-1}J\big)\e^{-itH_0}\varphi_0\big\rangle_\H\big|
&=\big|\big\langle T(z)\varphi,T_0 (z)\e^{-itH_0}\varphi_0\big\rangle_\G\big|\\
&\le\big\|T(z)\varphi\big\|_\G\;\!\big\|T_0(z)\e^{-itH_0}\varphi_0\big\|_\G.
\end{align*}
Therefore, it is enough to prove that $\big\|T_0(z)\e^{-itH_0}\varphi_0\big\|_\G\to0$
as $|t|\to\infty$. But since $T_0(z)\e^{-itH_0}\varphi_0$ and its derivative are
square integrable in $t$, this follows from a standard Sobolev embedding argument. So,
the existence of the limits \eqref{eq_wave} has been established. Similar arguments,
using the relation
$$
(H_0-\overline z)^{-1}J^*-J^*(H-\overline z)^{-1}=T_0(z)^*T(z)
$$
instead of
$$
J(H_0-z)^{-1}-(H-z)^{-1}J=T(z)^*T_0(z),
$$
show that $W_\pm(H_0,H,J^*,I)$ exists too. This, together with standard arguments in
scattering theory, implies the claims that follow \eqref{eq_wave}.
\end{proof}

As a consequence of Theorem \ref{thm_spec_H}(a) \& Theorem \ref{thm_wave}, we obtain
the following criterion for the existence and completeness of the local wave
operators\;\!:

\begin{Corollary}\label{corol_wave}
Let $H_0,H$ be self-adjoint operators in Hilbert spaces $\H_0,\H$ with spectral
measures $E^{H_0},E^H$ and $A_0,A$ self-adjoint operators in $\H_0,\H$. Assume that
$H_0,H$ are of class $C^{1+\varepsilon}(A_0),C^{1+\varepsilon}(A)$ for some
$\varepsilon\in(0,1)$. Let
$$
I:=\big\{\widetilde\mu^{A_0}(H_0)\setminus\sigma_{\rm p}(H_0)\big\}
\cap\big\{\widetilde\mu^A(H)\setminus\sigma_{\rm p}(H)\big\},
$$
$J\in\B(\H_0,\H)$, $\G$ an auxiliary Hilbert space, and for each $z\in\C\setminus\R$
suppose there exist $T_0(z)\in\B(\H_0,\G)$ and $T(z)\in\B(\H,\G)$ with
\begin{equation}\label{eq_decomposition}
J(H_0-z)^{-1}-(H-z)^{-1}J=T(z)^*T_0(z)
\end{equation}
and such that $T_0(z)$ extends continuously to an element of
$\B\big(\dom(\langle A_0\rangle^s)^*,\G\big)$ and $T(z)$ extends continuously to an
element of $\B\big(\dom(\langle A\rangle^s)^*,\G\big)$ for some $s>1/2$. Then the
local wave operators
$$
W_\pm(H,H_0,J,I)=\slim_{t\to\pm\infty}\e^{itH}J\e^{-itH_0}E^{H_0}(I)
$$
exist, are $J$-complete on $I$, and satisfy the relations
$$
W_\pm(H,H_0,J,I)^*=W_\pm(H_0,H,J^*,I)
\quad\hbox{and}\quad
W_\pm(H,H_0,J,I)\;\!\eta(H_0)=\eta(H)\;\!W_\pm(H,H_0,J,I)
$$
for each bounded Borel function $\eta:\R\to\C$.
\end{Corollary}

\subsection{Scattering theory for one-dimensional coupled photonic crystals}
\label{sec_scatt_photonic}

In the case of the pair $(M_0,M)$, we obtain the following result on the existence
and completeness of the wave operators; we use the notation $E_{\rm ac}^M$ for the
orthogonal projection on the absolutely continuous subspace of $M:$

\begin{Theorem}\label{thm_wave_max}
Let $I_{\rm max}:=\sigma(M_0)\setminus\{\Tau_M\cup\sigma_{\rm p}(M)\}$. Then
the local wave operators
$$
W_\pm(M,M_0,J,I_{\rm max})
:=\slim_{t\to\pm\infty}\e^{itM}J\e^{-itM_0}E^{M_0}(I_{\rm max})
$$
exist and satisfy
$\overline{\Ran\big(W_\pm(M,M_0,J,I_{\rm max})\big)}=E_{\rm ac}^M\;\!\H_w$. In
addition, the relations
$$
W_\pm(M,M_0,J,I_{\rm max})^*=W_\pm(M_0,M,J^*,I_{\rm max})
$$
and
$$
W_\pm(M,M_0,J,I_{\rm max})\;\!\eta(M_0)=\eta(M)\;\!W_\pm(M,M_0,J,I_{\rm max})
$$
hold for each bounded Borel function $\eta:\R\to\C$.
\end{Theorem}

\begin{proof}
All the claims except the equality
$
\overline{\Ran\big(W_\pm(M,M_0,J,I_{\rm max})\big)}=E_{\rm ac}^M\;\!\H_w
$
follow from Corollary \ref{corol_wave} whose assumptions are verified below.

Let $I\subset\sigma(M_0)\setminus\{\Tau_M\cup\sigma_{\rm p}(M)\}$ be a compact
interval. Then we know from Section \ref{sec_conj_free} that $M_0$ is of class
$C^2(A_{0,I})$ and from Section \ref{sec_spec} that $M$ is of class
$C^{1+\varepsilon}(A_I)$ for some $\varepsilon\in(0,1)$. Moreover, Theorems
\ref{thm_Mourre_Free} \& \ref{thm_regul} imply that
$$
I\subset\widetilde\mu^{A_{0,I}}(M_0)
\cap\big\{\widetilde\mu^{A_I}(M)\setminus\sigma_{\rm p}(M)\big\}.
$$
Therefore, in order to apply Corollary \ref{corol_wave}, it is sufficient to prove
that  for any $z\in\C\setminus\R$ the operator
$$
B(z)=J\left(M_0-z\right)^{-1}-(M-z)^{-1}J
$$
factorises as a product of two locally smooth operators as in \eqref{eq_decomposition}.
For that purpose, we set $s:=\tfrac{1+\widetilde\varepsilon}2$ with
$\widetilde\varepsilon\in(0,\varepsilon)$, we define
$$
\D:=\big\{\S(\R,\C^2)\oplus\S(\R,\C^2)\big\}\times\S(\R,\C^2)\subset\H_0\times\H_w,
$$
and we consider the sesquilinear form
\begin{equation}\label{eq_sesq_0}
\D\ni\big((\varphi_\ell,\varphi_{\rm r}),\varphi\big)\mapsto
\big\langle\langle Q\rangle^s\varphi,
B(z)\big(\langle Q_\ell\rangle^s\varphi_\ell,\langle Q_{\rm r}\rangle^s\varphi_{\rm r}\big)
\big\rangle_{\H_w}\in\C.
\end{equation}
Our first goal is to show that this sesquilinear form is continuous for the topology
of $\H_0\times \H_w$. However, since the necessary computations are similar to the
ones presented in Sections \ref{sec_conj_full}-\ref{sec_spec}, we only sketch them.
We know from \eqref{eq_two_terms} that
$$
B(z)\big(\langle Q_\ell\rangle^s\varphi_\ell,
\langle Q_{\rm r}\rangle^s\varphi_{\rm r}\big)
=\sum_{\star\in\{\ell,{\rm r}\}}\big\{\big((M_\star-z)^{-1}-(M-z)^{-1}\big)
\;\!j_\star\langle Q_\star\rangle^s\varphi_\star
+[j_\star,(M_\star-z)^{-1}]\langle Q_\star\rangle^s\varphi_\star\big\}.
$$
So, we have to establish the continuity of the sesquilinear forms
\begin{equation}\label{eq_sesq_1}
\S(\R,\C^2)\times\S(\R,\C^2)\ni(\varphi_\star,\varphi)\mapsto
\big\langle\langle Q\rangle^s\varphi,\big((M_\star-z)^{-1}-(M-z)^{-1}\big)
\;\!j_\star\;\!\langle Q_\star\rangle^s\varphi_\star\big\rangle_{\H_w}\in\C
\end{equation}
and
\begin{equation}\label{eq_sesq_2}
\S(\R,\C^2)\times\S(\R,\C^2)\ni(\varphi_\star,\varphi)\mapsto
\big\langle\langle Q\rangle^s\varphi,[j_\star,(M_\star-z)^{-1}]
\langle Q_\star\rangle^s\varphi_\star\big\rangle_{\H_w}\in\C.
\end{equation}
For the first one, we know from \eqref{eq_pre_trick} that
\begin{align}
&\big((M_\star-z)^{-1}-(M-z)^{-1}\big)\;\!j_\star\;\!\langle Q_\star\rangle^s
\nonumber\\
&=(M-z)^{-1}j_\star(w-w_\star)D(M_\star-z)^{-1}\langle Q_\star\rangle^s
+(M-z)^{-1}(w-w_\star)[D,j_\star](M_\star-z)^{-1}\langle Q_\star\rangle^s\nonumber\\
&\quad+(M-z)^{-1}(M-M_\star)[(M_\star-z)^{-1},j_\star]
\langle Q_\star\rangle^s.\label{eq_x1}
\end{align}
By inserting this expression into \eqref{eq_sesq_1}, by taking the
$C^1(\langle Q\rangle^\alpha)$-property of $M$ and $M_\star$ into account, and by
observing that the operators $[D,\langle Q_\star\rangle^s]$,
$\langle Q\rangle^sj_\star(w-w_\star)\langle Q_\star\rangle^s$ and
$\langle Q\rangle^s[D,j_\star]\langle Q_\star\rangle^s$ defined on $\S(\R,\C^2)$
extend continuously to elements of $\B(\H_w)$, one obtains that the sesquilinear
forms defined by the first two terms in \eqref{eq_x1} are continuous for the topology
of $\H_{w_\star}\times\H_w$. The sesquilinear form defined by the third term in
\eqref{eq_x1} and the sesquilinear form \eqref{eq_sesq_2} can be treated
simultaneously. Indeed, the factor $[j_\star,(M_\star-z)^{-1}]$ can be computed
explicitly and contains a factor $j'_\star$ which has compact support. Therefore,
since $\langle Q\rangle^sj'_\star\langle Q_\star\rangle^s\in\B(\H_w)$, a few more
commutator computations show that the two remaining sesquilinear forms are continuous
for the topology of $\H_{w_\star}\times\H_w$.

In consequence, the sesquilinear form \eqref{eq_sesq_0} is continuous for the topology
of $\H_0\times \H_w$, and thus corresponds to a bounded operator
$F_z\in\B(\H_0,\H_w)$. Therefore, if we set
$$
T_0(z):=\langle Q_\ell\rangle^{-s}\oplus\langle Q_{\rm r}\rangle^{-s}\in\B(\H_0)
\quad\hbox{and}\quad
T(z):=F_z^*\langle Q\rangle^{-s}\in\B(\H_w,\H_0),
$$
we obtain that $B(z)=T(z)^*T_0(z)$. On another hand, we know from computations
presented in Section~\ref{sec_spec} that
$$
\langle Q\rangle^{-s}\in\B\big(\dom(\langle Q\rangle^s)^*,\H_w\big)
\subset\B\big(\dom(\langle A_I\rangle^s)^*,\H_w\big),
$$
and
$$
\langle Q_\ell\rangle^{-s}\oplus\langle Q_{\rm r}\rangle^{-s}
\in\B\big(\dom(\langle Q_\ell\rangle^s\oplus\langle Q_{\rm r}\rangle^s)^*,\H_0\big)
\subset\B\big(\dom(\langle A_{0,I}\rangle^s)^*,\H_0\big).
$$
So, we have the inclusions
$$
T(z)\in\B\big(\dom(\langle A_I\rangle^s)^*,\H_0\big)
\quad\hbox{and}\quad
T_0(z)\in\B\big(\dom(\langle A_{0,I}\rangle^s)^*,\H_0\big),
$$
and thus all the assumptions of Corollary \ref{corol_wave} are verified.

Hence it only remains to show that
$\overline{\Ran\big(W_\pm(M,M_0,J,I_{\rm max})\big)}=E_{\rm ac}^M\;\!\H_w$. For that
purpose, we first recall from the proof of Theorem \ref{thm_regul} that
$E^M(I)(JJ^*-1)E^M(I)\in\K(\H_w)$.
Then since $M$ has purely absolutely continuous spectrum in $I$ one infers from the
RAGE theorem that
$$
\slim_{t\to\pm\infty}E^M(I)\e^{itM}(JJ^*-1)\e^{-itM} E^M(I)=0,
$$
and consequently that $\NN_\pm(M,J,I)=E^M(\R\setminus I)\H_w$. By using the
$J$-completeness on $I$ of the local wave operators and that $M$ has purely absolutely
continuous spectrum in $I$, we thus obtain
$$
\overline{\Ran\big(W_\pm(M,M_0,J,I)\big)}
=E^M(I)\;\!\H\ominus\NN_\pm(M,J,I)
=E^M(I)\;\!\H_w
=E^M(I)\;\!\H_w\cap E^M_{\rm ac}\;\!\H_w.
$$
By putting together these results for different intervals $I$ and by using Proposition
\ref{proposition_ess_M}, we thus get that
\begin{align*}
\overline{\Ran\big(W_\pm(M,M_0,J,I_{\rm max})\big)}
&=E^M(I_{\rm max})\;\!\H_w\cap E_{\rm ac}^M\;\!\H_w\\
&=E^M\big(\sigma_{\rm ess}(M)\big)\;\!\H_w\cap E_{\rm ac}^M\;\!\H_w\\
&=E_{\rm ac}^M\;\!\H_w,
\end{align*}
which concludes the proof.
\end{proof}

\begin{Remark}\label{rem_sum_wave}
Let $I\subset\sigma(M_0)\setminus\{\Tau_M\cup\sigma_{\rm p}(M)\}$ be a compact
interval and let $(\varphi_\ell,\varphi_{\rm r})\in\H_0$. Then we have
\begin{align*}
W_\pm(M,M_0,J,I)(\varphi_\ell,\varphi_{\rm r})
&=\slim_{t\to\pm\infty}\e^{itM}J\e^{-itM_0}E^{M_0}(I)(\varphi_\ell,\varphi_{\rm r})\\
&=\slim_{t\to\pm\infty}\e^{itM}\big(J_\ell\e^{-itM_\ell}E^{M_\ell}(I)\varphi_\ell
+J_{\rm r}\e^{-itM_{\rm r}}E^{M_{\rm r}}(I)\varphi_{\rm r}\big)\\
&=W_\pm(M,M_\ell,J_\ell,I)\;\!\varphi_\ell
+W_\pm(M,M_{\rm r},J_{\rm r},I)\;\!\varphi_{\rm r}
\end{align*}
with
\begin{equation}\label{partial_wave_op}
W_\pm(M,M_\star,J_\star,I)
:=\slim_{t\to\pm\infty}\e^{itM}J_\star\e^{-itM_\star}E^{M_\star}(I)
\end{equation}
and $J_\star\in\B(\H_{w_\star},\H_w)$ given by
\begin{equation}\label{def_J_star}
J_\star\varphi_\star:=j_\star\;\!\varphi_\star,
\quad\varphi_\star\in\H_{w_\star}.
\end{equation}
That is, the operators $W_\pm(M,M_0,J,I)$ act as the sum of the local wave operators
$W_\pm(M,M_\star, J_\star,I):$
$$
W_\pm(M,M_0,J,I)(\varphi_\ell,\varphi_{\rm r})
=W_\pm(M,M_\ell,J_\ell,I)\;\!\varphi_\ell
+W_\pm(M,M_{\rm r},J_{\rm r},I)\;\!\varphi_{\rm r}.
$$
\end{Remark}

In order to get a better understanding of the initial sets of the partial isometries
$W_\pm(M,M_0,J,I_{\rm max})$ some preliminary considerations on the asymptotic
velocity operators for $M_\ell$ and $M_{\rm r}$ are necessary. First, we define for
each $\star\in\{\ell,{\rm r}\}$ and $n\in\N$ the spaces
$$
\H_{\star,n}:=\widehat\Pi_{\star,n}\;\!\H_{\tau,\star}
\quad\hbox{and}\quad
\H_{\star,n}^\infty
:=\widehat\Pi_{\star,n}\;\!\big\{\H_{\tau,\star}\cap C^\infty(\R,\h_\star)\big\},
$$
and note that $\H_{\tau,\star}$ decomposes into the internal direct sum
$\H_{\tau,\star}=\oplus_{n\in\N}\H_{\star,n}$ and that the operator
$\widehat{M_\star}$ is reduced by this decomposition, namely,
$\widehat{M_\star}=\sum_{n\in\N}\widehat\lambda_{\star,n}\widehat\Pi_{\star,n}$. Next,
we introduce the self-adjoint operator $\widehat V_\star$ in $\H_{\tau,\star}$
$$
\widehat V_\star
:=\sum_{n\in\N}\widehat\lambda'_{\star,n}\widehat\Pi_{\star,n},
\quad\dom\big(\widehat V_\star\big):=\left\{u\in\H_{\tau,\star}\mid
\big\|\widehat V_\star u\big\|_{\H_{\tau,\star}}^2
=\sum_{n\in\N}\big\|\widehat\lambda'_{\star,n}
\widehat\Pi_{\star,n}u\big\|_{\H_{\tau,\star}}^2<\infty\right\}.
$$
Then it is natural to define the asymptotic velocity operator $V_\star$ for $M_\star$
in $\H_{w_\star}$ as
$$
V_\star:=\U^{-1}_\star\widehat V_\star\U_\star,
\quad\dom(V_\star):=\U_\star^{-1}\dom\big(\widehat V_\star\big),
$$
and the asymptotic velocity operator $V_0$ for $M_0$ in $\H_0$ as the direct sum
$$
V_0:=V_\ell\oplus V_{\rm r}.
$$
Additionally, we define the family of self-adjoint operators in $\H_{w_\star}$
$$
Q_\star(t):=\e^{itM_\star}Q_\star\e^{-itM_\star},
\quad t\in\R,~\dom\big(Q_\star(t)\big):=\e^{itM_\star}\dom(Q_\star),
$$
and the corresponding family of self-adjoint operators in $\H_0$
$$
Q_0(t):=Q_\ell(t)\oplus Q_{\rm r}(t),\quad t\in\R.
$$

Our next result is inspired by the result of \cite[Thm.~4.1]{Suz16} in the setup of
quantum walks. In the proof, we use the linear span $\H_{\star,\tau}^{\rm fin}$ of
elements of $\H_{\star,n}^\infty:$
$$
\H_{\star,\tau}^{\rm fin}
:=\left\{\sum_{n\in\N}u_n\in\bigoplus_{n\in\N}\H_{\star,n}^\infty
\mid\hbox{$u_n\ne0$ for only a finite number of $n$}\right\}.
$$

\begin{Proposition}\label{prop_V_star}
For each $\star\in\{\ell,{\rm r}\}$ and $z\in\C\setminus\R$, we have
$$
\slim_{t\to\pm\infty}\left(\tfrac{Q_\star(t)}t-z\right)^{-1}=\big(V_\star-z\big)^{-1}.
$$
\end{Proposition}

\begin{proof}
For each $t\in\R$, we have the inclusion
$
\U_\star^{-1}\H_{\star,\tau}^{\rm fin}
\subset\big\{\dom(V_\star)\cap\dom\big(Q_\star(t)\big)\big\}
$.
Furthermore, if $u\in\H_{\star,\tau}^{\rm fin}$, then we have
$$
\big(V_\star-z\big)^{-1}\U_\star^{-1}u
=\U_\star^{-1}\big(\widehat V_\star-z\big)^{-1}u
\in\U_\star^{-1}\H_{\star,\tau}^{\rm fin}.
$$
As a consequence, the following equality holds for all $t\in\R\setminus\{0\}$ and
$u\in\H_{\star,\tau}^{\rm fin}:$
$$
\left(\left(\tfrac{Q_\star(t)}t-z\right)^{-1}
-\big(V_\star-z\big)^{-1}\right)\U_\star^{-1}u
=\left(\tfrac{Q_\star(t)}t-z\right)^{-1}\left(V_\star-\tfrac{Q_\star(t)}t\right)
\big(V_\star-z\big)^{-1}\U_\star^{-1}u.
$$
Since
$
\big\|\big(\frac{Q_\star(t)}t-z\big)^{-1}\big\|_{\B(\H_{w_\star})}\le|\im(z)|^{-1}
$
and $\big\|(V_\star-z)^{-1}\big\|_{\B(\H_{w_\star})}\le |\im(z)|^{-1}$, and since
$\U^{-1}_\star\H_{\star,\tau}^{\rm fin}$ is dense in $\H_{w_\star}$, it follows that
it is sufficient to prove that
$$
\lim_{t\to\pm\infty}\left\|\left(V_\star-\tfrac{Q_\star(t)}t\right)
\varphi_\star\right\|_{\H_{w_\star}}=0
\quad\hbox{for all $\varphi_\star\in\U_\star^{-1}\H_{\star,\tau}^{\rm fin}$.}
$$
Now, a direct calculation using the Bloch-Floquet transform gives for
$\varphi_\star=\U_\star^{-1}u$ with $u\in\H_{\star,\tau}^{\rm fin}$
\begin{align*}
&\left\|\left(V_\star-\tfrac{Q_\star(t)}t\right)\varphi_\star\right\|^2_{\H_{w_\star}}\\
&=\int_{Y^*_\star}\d k\left\|\sum_{n\in\N}
\big(\widehat\lambda'_{\star,n}\widehat\Pi_{\star,n}u\big)(k)
-\left(\e^{it\widehat M_\star}\tfrac{\widehat{Q_\star}}t\sum_{n\in\N}
\e^{-it\widehat\lambda_{\star,n}}
\widehat\Pi_{\star,n}u\right)(k)\right\|^2_{\h_\star}\\
&=\tfrac1{t^2}\int_{Y^*_\star}\d k\left\|\sum_{n\in\N}
\big(\widehat{Q_\star}\widehat\Pi_{\star,n}u\big)(k)\right\|^2_{\h_\star},
\end{align*}
where in the last equation we have used that $\widehat{Q_\star}$ acts as $i\partial_k$
in $\H_{\tau,\star}$. Since $u\in\H_{\star,\tau}^{\rm fin}$, the summation over
$n\in\N$ is finite, and since the map
$Y^*_\star\ni k\mapsto\big(\widehat{Q_\star}\widehat\Pi_{\star,n}u\big)(k)\in\h_\star$
is bounded, one deduces that
$$
\left\|\left(V_\star-\tfrac{Q_\star(t)}t\right)\varphi_\star\right\|_{\H_{w_\star}}
=\O\big(t^{-1}\big),
$$
which implies the claim.
\end{proof}

In the next proposition, we determine the initial sets of the isometries
$W_\pm(M,M_\star,J_\star,I):\H_{w_\star}\to\H_w$ defined in \eqref{partial_wave_op}.
In the statement, we use the fact that the operators $M_\star$ and $V_\star$ strongly
commute. We also use the notations $\chi_+$ and $\chi_-$ for the characteristic
functions of the intervals $(0,\infty)$ and $(-\infty,0)$, respectively.

\begin{Proposition}\label{prop_initial_lr}
\begin{enumerate}
\item[(a)] Let $I\subset\sigma(M_\ell)\setminus\Tau_\ell$ be a compact interval, then
the operators $W_\pm(M,M_\ell,J_\ell,I):\H_{w_\ell}\to\H_w$ are partial isometries
with initial sets $\chi_\mp(V_\ell)E^{M_\ell}(I)\H_{w_\ell}$.
\item[(b)] Let $I\subset\sigma(M_{\rm r})\setminus\Tau_{\rm r}$ be a compact interval,
then the operators $W_\pm(M,M_{\rm r},J_{\rm r},I):\H_{w_{\rm r}}\to\H_w$ are partial
isometries with initial sets $\chi_\pm(V_{\rm r})E^{M_{\rm r}}(I)\H_{w_{\rm r}}$.
\end{enumerate}
\end{Proposition}

Before the proof, let us observe that if $I\subset\sigma(M_\star)\setminus\Tau_\star$
is a compact interval, then we have the equalities
\begin{equation}\label{eq_prod_proj}
\chi_-(V_\star)E^{M_\star}(I)=\chi_{(-\infty,0]}(V_\star)E^{M_\star}(I)
\quad\hbox{and}\quad
\chi_+(V_\star)E^{M_\star}(I)=\chi_{[0,+\infty)}(V_\star)E^{M_\star}(I)
\end{equation}
due to the definition of the set $\Tau_\star$.

\begin{proof}
Our proof is inspired by the proof of \cite[Prop.~3.4]{RST_2}. We first show the claim
for $W_+(M,M_\ell,J_\ell,I)$. So, let $\varphi_\ell\in\H_{w_\ell}$. If
$\varphi_\ell\bot E^{M_\ell}(I)\H_{w_\ell}$, then
$\varphi_\ell\in\ker\big(W_+(M,M_\ell,J_\ell,I)\big)$. Thus, we can assume that
$\varphi_\ell\in E^{M_\ell}(I)\H_{w_\ell}$. Next, let us show that if
$\varphi_\ell\in\chi_+(V_\ell)\H_{w_\ell}$ then again
$\varphi_\ell\in\ker\big(W_+(M,M_\ell,J_\ell,I)\big)$. For this, assume that
$\chi_{[\varepsilon,\infty)}(V_\ell)\varphi_\ell=\varphi_\ell$ for some
$\varepsilon>0$. Then it follows from \eqref{partial_wave_op}-\eqref{def_J_star} that
\begin{align*}
\big\|W_+(M,M_\ell,J_\ell,I)\varphi_\ell\big\|_{\H_w}
&=\left\|\slim_{t\to+\infty}\e^{itM}J_\ell\e^{-itM_\ell}\varphi_\ell\right\|_{\H_w}\\
&=\lim_{t\to+\infty}\left\|\e^{itM}J_\ell\e^{-itM_\ell}\varphi_\ell\right\|_{\H_w}\\
&\le{\rm Const.}\lim_{t\to+\infty}
\big\|\e^{itM_\ell}j_\ell\e^{-itM_\ell}\varphi_\ell\big\|_{\H_{w_\ell}}.
\end{align*}
Now, let $\eta_\ell\in C(\R,[0,1])$ satisfy $\eta_\ell(x)=1$ if $x<0$ and
$\eta_\ell(x)=0$ if $x\ge\varepsilon$. Then one has for each $t>0$ the inequality
$$
\big\|\e^{itM_\ell}j_\ell\e^{-itM_\ell}\varphi_\ell\big\|_{\H_{w_\ell}}
\le\big\|\e^{itM_\ell}\eta_\ell(Q_\ell/t)\e^{-itM_\ell}
\varphi_\ell\big\|_{\H_{w_\ell}}.
$$
Furthermore, since
$
\eta_\ell(V_\ell)\varphi_\ell
=\eta_\ell(V_\ell)\;\!\chi_{[\varepsilon,\infty)}(V_\ell)\varphi_\ell=0
$,
one infers from Proposition \ref{prop_V_star} and from a standard result on strong
resolvent convergence \cite[Thm.~VIII.20(b)]{RS1} that
$$
\lim_{t\to+\infty}
\big\|\e^{itM_\ell}\eta_\ell(Q_\ell/t)\e^{-itM_\ell}\varphi_\ell\big\|_{\H_{w_\ell}}
=\big\|\eta_\ell(V_\ell)\varphi_\ell\big\|_{\H_{w_{\ell}}}
=0.
$$
Putting together what precedes, one obtains that
$
\varphi_\ell
=\chi_{[\varepsilon,\infty)}(V_\ell)\varphi_\ell
\in\ker\big(W_+(M,M_\ell,J_\ell,I)\big)
$,
and then a density argument taking into account the second equation in
\eqref{eq_prod_proj} implies that
$$
\chi_+(V_\ell)\H_{w_\ell}\subset\ker\big(W_+(M,M_\ell,J_\ell,I)\big).
$$

To show that $W_+(M,M_\ell,J_\ell,I)$ is an isometry on
$\chi_-(V_\ell)E^{M_\ell}(I)\H_{w_\ell}$, take
$\varphi_\ell\in\chi_{-}(V_\ell)E^{M_\ell}(I)\H_{w_\ell}$ with
$\chi_{(-\infty,-\varepsilon]}(V_\ell)\varphi_\ell=\varphi_\ell$ for some
$\varepsilon>0$, and let $\zeta_\ell\in C(\R,[0,1])$ satisfy $\zeta_\ell(x)=0$ if
$x\le-\varepsilon$ and $\zeta_\ell(x)=1$ if $x>-\varepsilon/2$. Then using
successively the identity $E^{M_\ell}(I)\varphi_\ell=\varphi_\ell$, the unitarity of
$\e^{itM}$ in $\H_w$ and of $\e^{-itM_\ell}$ in $\H_{w_\ell}$, the definition
\eqref{def_J_star} of $J_\ell$, the definition of $V_\ell$, and the fact that
$\chi_{(-\infty,-\varepsilon]}(V_\ell)\varphi_\ell=\varphi_\ell$, one gets
\begin{align*}
& \left|\big\|W_+(M,M_\ell,J_\ell,I)\varphi_\ell\big\|_{\H_w}^2
-\|\varphi_\ell\|_{\H_{w_\ell}}^2\right| \\
&=\lim_{t\to+\infty}\left|\big\|\e^{itM}J_\ell\e^{-itM_\ell}\varphi_\ell\big\|_{\H_w}^2
-\|\varphi_\ell\|_{\H_{w_\ell}}^2\right|\\
&=\lim_{t\to+\infty}\left|\big\|J_\ell\e^{-itM_\ell}\varphi_\ell\big\|_{\H_w}^2
-\big\|\e^{-itM_\ell}\varphi_\ell\big\|_{\H_{w_\ell}}^2\right|\\
&=\lim_{t\to+\infty}\left|\big\langle\e^{-itM_\ell}\varphi_\ell,
\big(1-w_\ell w^{-1}j_\ell^2\big)\e^{-itM_\ell}\varphi_\ell
\big\rangle_{\H_{w_\ell}}\right|\\
&\le\lim_{t\to+\infty}\big\langle\varphi_\ell,
\e^{itM_\ell}(1-j_\ell^2)\e^{-itM_\ell}\varphi_\ell\big\rangle_{\H_{w_\ell}}
+\lim_{t\to+\infty}\left|\big\langle\varphi_\ell,
\e^{itM_\ell}(w-w_\ell)\;\!j_\ell^2 w^{-1}\e^{-itM_\ell}\varphi_\ell
\big\rangle_{\H_{w_\ell}}\right|.
\end{align*}
For the first term one has
\begin{align*}
\lim_{t\to+\infty}\big\langle\varphi_\ell,
\e^{itM_\ell}(1-j_\ell^2)\e^{-itM_\ell}\varphi_\ell\big\rangle_{\H_{w_\ell}} 
&\le\lim_{t\to+\infty}\big\langle\varphi_\ell,\e^{itM_\ell}\zeta_\ell(Q_\ell/t)
\e^{-itM_\ell}\varphi_\ell\big\rangle_{\H_{w_\ell}}\\
&=\big\langle\varphi_\ell,\zeta_\ell(V_\ell)\varphi_\ell\big\rangle_{\H_{w_\ell}}\\
&=0,
\end{align*}
while the second term also vanishes by an application of the RAGE theorem. It follows
that $W_+(M,M_\ell,J_\ell,I)$ is isometric on
$\varphi_\ell=\chi_{(-\infty,-\varepsilon]}(V_\ell)\varphi_\ell$, and then a density
argument taking into account the first equation in \eqref{eq_prod_proj} implies that
$W_+(M,M_\ell,J_\ell,I)$ is isometric on $\chi_-(V_\ell)E^{M_\ell}(I)\H_{w_\ell}$.

A similar proof works for the claims about $W_-(M,M_\ell,J_\ell,I)$ and
$W_\pm(M,M_{\rm r},J_{\rm r},I)$. The functions $\eta_\ell$ and $\zeta_\ell$ have to
be adapted and the possible negative sign of the variable $t$ has to be taken into
account, otherwise the argument can be copied \it{mutatis mutandis}.
\end{proof}

By collecting the results of Theorem \ref{thm_wave_max}, Remark \ref{rem_sum_wave},
Proposition \ref{prop_initial_lr}, and by using the fact that $M_\ell$ and $M_{\rm r}$
have purely absolutely continuous spectrum, one finally obtains a description of the
initial sets of the local wave operators $W_\pm(M,M_0,J,I_{\rm max}):$

\begin{Theorem}\label{thm_initial_max}
Let $I_{\rm max}:=\sigma(M_0)\setminus\{\Tau_M\cup\sigma_{\rm p}(M)\}$ and
$I_\star:=\sigma(M_\star)\setminus\Tau_\star$ ($\star=\ell,\rm r$). Then the local
wave operators $W_\pm(M,M_0,J,I_{\rm max}):\H_0\to\H_w$ are partial isometries with
initial sets
$$
\H_0^\pm:=\chi_\mp(V_\ell)E^{M_\ell}(I_\ell)\H_{w_\ell}
\oplus\chi_\pm(V_{\rm r})E^{M_{\rm r}}(I_{\rm r})\H_{w_{\rm r}}.
$$
\end{Theorem}

\begin{Remark}
One has $E^{M_\star}(I_\star)=1$ because $\Tau_\star$ is discrete and
$\sigma_{\rm p}(M_\star)=\varnothing$ (see Proposition \ref{prop_spec_asymp} and the
paragraph that precedes). Therefore, the spectral projections $E^{M_\star}(I_\star)$
in the statement of Theorem \ref{thm_initial_max} can be removed if desired.
\end{Remark}



\end{document}